\documentclass[11pt]{article}
\usepackage{microtype}
\usepackage{fullpage,amsthm,amssymb,amsmath}
\usepackage{graphicx,algorithmic,algorithm}%
\usepackage{enumerate}
\usepackage{tikz}
\usetikzlibrary{shapes}
\usepackage{hyperref}
\usepackage[normalem]{ulem}
\hypersetup{
    pdftitle=   {A width parameter useful for chordal and co-comparability graphs},
   pdfauthor=  {Dong Yeap Kang and O-joung Kwon and Torstein J.\,F. Str{\o}mme and Jan Arne Telle}
}
\usepackage{authblk}
\newtheorem{theorem}{Theorem}[section]
\newtheorem{lemma}[theorem]{Lemma}
\newtheorem{proposition}[theorem]{Proposition}
\newtheorem{corollary}[theorem]{Corollary}

\newtheorem{claim}{Claim}
\newtheorem{QUE}{Question}
\newtheorem*{thm:main}{Theorem~\ref{thm:main}}
\newenvironment{clproof}{\begin{list}{}{%
              \setlength{\leftmargin}{5mm}%
              } \item {\it Proof.} }{\hfill$\lozenge$\end{list}\medskip}

\newcommand\abs[1]{\lvert #1\rvert}

\newcommand\cutrk{\operatorname{cutrk}}

\newcommand\tw{\operatorname{tw}}

\newcommand\rw{\operatorname{rw}}

\newcommand\simw{\operatorname{simw}}
\newcommand\lsimw{\operatorname{lsimw}}
\newcommand\lmimw{\operatorname{lmimw}}
\newcommand\simval{\operatorname{sim}}

\newcommand\mimw{\operatorname{mimw}}
\newcommand\mimval{\operatorname{mim}}

\newcommand\mat{\boxminus}

\newcommand\cB{\mathcal B}

\usepackage{xspace}
\newcommand{\ie}{i.\,e.\@\xspace}

\begin{document}
\title{A width parameter useful for chordal and co-comparability graphs}

\author[1]{Dong Yeap Kang\thanks{Supported by TJ Park Science Fellowship of POSCO TJ Park Foundation.}}

\affil[1]{Department of Mathematical Sciences, KAIST,  Daejeon, South Korea.}

\author[2]{O-joung Kwon\thanks{Supported by the European Research Council (ERC) under the European Union's Horizon 2020 research and innovation programme (ERC consolidator grant DISTRUCT, agreement No. 648527).}}

\affil[2]{Logic and Semantics, Technische Universit\"at Berlin, Berlin, Germany.}

\author[3]{Torstein J.\,F. Str{\o}mme}

\author[3]{Jan Arne Telle}
\affil[3]{Department of Informatics, University of Bergen, Norway.}

\date\today
\maketitle

\footnotetext{E-mail addresses: \texttt{dynamical@kaist.ac.kr} (D. Kang), \texttt{ojoungkwon@gmail.com} (O. Kwon), \texttt{torstein.stromme@ii.uib.no} (T. Str{\o}mme), 
\texttt{Jan.Arne.Telle@uib.no} (J. A. Telle). \\
An extended abstract appeared in the proceedings of 11th International Workshop on Algorithms and Computation, 2017. }

\begin{abstract}

Belmonte and Vatshelle (TCS 2013) used mim-width, a graph width parameter bounded on interval graphs and permutation graphs, to explain existing algorithms for many domination-type problems on those graph classes.
We investigate new graph classes of bounded mim-width, strictly extending interval graphs and permutation graphs.
The graphs $K_t \mat K_t$ and $K_t \mat S_t$ are graphs obtained from the disjoint union of two cliques of size $t$, and one clique of size $t$ and one independent set of size $t$ respectively, by adding a perfect matching.
We prove that :
\begin{itemize}
\item interval graphs are $(K_3\mat S_3)$-free chordal graphs; and $(K_t \mat S_t)$-free chordal graphs have mim-width at most $t-1$, 
\item permutation graphs are $(K_3\mat K_3)$-free co-comparability graphs; and $(K_t \mat K_t)$-free co-comparability graphs have mim-width at most $t-1$, 
\item chordal graphs and co-comparability graphs have unbounded mim-width in general.
\end{itemize}
We obtain several algorithmic consequences; for instance, while \textsc{Minimum Dominating Set} is NP-complete on chordal graphs, 
 it can be solved in time $n^{\mathcal{O}(t)}$ on $(K_t\mat S_t)$-free chordal graphs. 
 The third statement strengthens a result of Belmonte and Vatshelle stating that either those classes 
 do not have constant mim-width or a 
 decomposition with constant mim-width cannot be computed in polynomial time unless $P=NP$.

We generalize these ideas to bigger graph classes. 
We introduce a new width parameter \emph{sim-width}, of stronger modelling power than mim-width, by making a small change in the definition of mim-width. We prove that chordal graphs and co-comparability graphs have sim-width at most 1.
We investigate a way to bound mim-width for graphs of bounded sim-width by excluding $K_t\mat K_t$ and $K_t\mat S_t$ as induced minors or induced subgraphs, and give algorithmic consequences.
Lastly, we show that circle graphs have unbounded sim-width, and thus also unbounded mim-width.
\end{abstract}

\section{Introduction}\label{sec:introduction}

The study of structural graph ``width'' parameters like tree-width and clique-width have been ongoing since at least the 1990's, and their algorithmic use has been steadily increasing~\cite{Cygan15}.
A parallel and somewhat older field of study gives algorithms for special graph classes, such as chordal graphs and permutation graphs. Since the introduction of algorithms based on tree-width, which generalized algorithms for trees and series-parallel graphs, these two fields have been connected. Recently the fields became even more intertwined with the introduction of mim-width in 2012~\cite{VatshelleThesis}, which yeilded generalized algorithms for quite large graph classes.
In the context of parameterized complexity a negative relationship holds between the modeling power of graph width parameters, \ie{} what graph classes have bounded parameter values, and their analytical power, \ie{} what problems become FPT or XP.
For example, the family of graph classes of bounded width is strictly larger for clique-width than for tree-width, while under standard complexity assumptions the class of problems solvable in FPT time on a decomposition of bounded width is strictly larger for tree-width than for clique-width. For a parameter like mim-width its algorithmic use must therefore be carefully evaluated against its modelling power, which is much stronger than clique-width.

A common framework for defining graph width parameters is by branch decompositions over the vertex set. This is a natural hierarchical clustering of a graph $G$, represented as an unrooted binary tree $T$ with the vertices of $G$ at its leaves. Any edge of the tree defines a cut of $G$ given by the leaves of the two subtrees that result from removing the edge from $T$. Judiciously choosing a cut-function to measure the complexity of such cuts, or rather of the bipartite subgraphs of $G$ given by the edges crossing the cuts, this framework then defines a graph width parameter by a minmax relation, minimum over all trees and maximum over all its cuts. Restricting $T$ to a path with added leaves we get a linear variant.
The first graph parameter defined this way was carving-width~\cite{DBLP:journals/combinatorica/SeymourT94}, whose cut-function is simply the number of edges crossing the cut, and whose modelling power is weaker than tree-width. The carving-width of a graph $G$ is thus the minimum, over all such branch decomposition trees, of the maximum number of edges crossing a cut given by an edge of the tree.
Several graph width parameters have been defined this way. For example the mm-width~\cite{VatshelleThesis}, whose cut function is the size of the maximum matching and has modelling power equal to tree-width; and the rank-width~\cite{OS2004}, whose cut function is the GF[2]-rank of the adjacency matrix and has modelling power equal to clique-width. 

This framework was used by Vatshelle in 2012~\cite{VatshelleThesis} to define the parameter mim-width whose cut function is the size of the maximum induced matching of the graph crossing the cut. Note that low mim-width allows quite complex cuts. Carving-width one allows just a single edge, mm-width one a star graph, and rank-width one
a complete bipartite graph with some isolated vertices. In contrast, mim-width one allows any cut where the neighborhoods of the vertices in a color class can be ordered linearly w.r.t. inclusion. The modelling power of mim-width is much stronger than clique-width. Belmonte and Vatshelle showed that interval graphs and permutation graphs have linear mim-width at most one, and circular-arc graphs and trapezoid graphs have linear mim-width at most two~\cite{BelmonteV2013} \footnote{In \cite{BelmonteV2013}, all the related results are stated in terms of $d$-neighborhood equivalence, but in the proof, they actually gave a bound on mim-width or linear mim-width.}, whereas the clique-width of such graphs can be proportional to the square root of the number of vertices. With such strong modelling power it is clear that the analytical power of mim-width must suffer. The cuts in a decomposition of constant mim-width are too complex to allow FPT algorithms for interesting NP-hard problems. Instead, what we get is XP algorithms, for the class of LC-VSVP problems~\cite{BinhminhJM2013} - locally checkable vertex subset and vertex partitioning problems - defined in Section \ref{subsec:lcvsvp}. For classes of bounded mim-width this gives a common explanation for many classical results in the field of algorithms on special graph classes.

In this paper we extend these results on mim-width in several ways. Firstly, we show that chordal graphs and co-comparability graphs have unbounded mim-width, thereby answering a question of ~\cite{BelmonteV2013}\footnote{Our result appeared on arXiv June 2016. In August 2016 a similar result by Mengel, developed independently, also appeared on arXiv \cite{Mengel}.}, and giving evidence to the intuition that mim-width is useful for large graph classes having a linear structure rather than those having a tree-like structure. Secondly, by excluding certain subgraphs, like $K_t \mat S_t$ obtained from the disjoint union of a clique of size $t$ and an independent set of size $t$ by adding a perfect matching, we find subclasses of chordal graphs and co-comparability graphs for which we can nevertheless compute a bounded mim-width decomposition in polynomial time and thereby solve LC-VSVP problems. See Figure~\ref{fig:construction} for illustrations.  Note that already for $K_3\mat S_3$-free chordal graphs the tree-like structure allowed by mim-width is necessary, as they have unbounded linear mim-width. Thirdly, we introduce sim-width, of modeling power even stronger than mim-width, and start a study of its properties. 

\begin{figure}[tb]
 \tikzstyle{v}=[circle, draw, solid, fill=black, inner sep=0pt, minimum width=3pt]
  \centering
    \begin{tikzpicture}[scale=0.4]
      \foreach \x in {1,...,5} {
        \node [v]  (v\x) at(0,-\x){};
        \node [v]  (w\x) at (2,-\x){};
      }
      \draw (v1)--(v5);
      \draw(v1) [in=120,out=-120] to (v3);
      \draw(v2) [in=120,out=-120] to (v4);
      \draw(v3) [in=120,out=-120] to (v5);
      \draw(v1) [in=120,out=-120] to (v4);
      \draw(v2) [in=120,out=-120] to (v5);
      \draw(v1) [in=120,out=-120] to (v5);
      \foreach \x in {1,...,5} 
      \foreach \y in {1,...,5} {
    \ifnum \x=\y  \draw (v\x)--(w\y);  \fi
      }
    \end{tikzpicture}
    \qquad \qquad \qquad
    \begin{tikzpicture}[scale=0.4]
      \foreach \x in {1,...,5} {
        \node [v]  (v\x) at(0,-\x){};
        \node [v]  (w\x) at (2,-\x){};
      }
      \draw (v1)--(v5);
      \draw(v1) [in=120,out=-120] to (v3);
      \draw(v2) [in=120,out=-120] to (v4);
      \draw(v3) [in=120,out=-120] to (v5);
      \draw(v1) [in=120,out=-120] to (v4);
      \draw(v2) [in=120,out=-120] to (v5);
      \draw(v1) [in=120,out=-120] to (v5);
      \draw (w1)--(w5);
      \draw(w1) [in=60,out=-60] to (w3);
      \draw(w2) [in=60,out=-60] to (w4);
      \draw(w3) [in=60,out=-60] to (w5);
      \draw(w1) [in=60,out=-60] to (w4);
      \draw(w2) [in=60,out=-60] to (w5);
      \draw(w1) [in=60,out=-60] to (w5);
      \foreach \x in {1,...,5} 
      \foreach \y in {1,...,5} {
    \ifnum \x=\y  \draw (v\x)--(w\y);  \fi
      }
    \end{tikzpicture}
      \caption{$K_5\mat S_5$ and $K_5\mat K_5$.}
  \label{fig:construction}
\end{figure}

The graph width parameter sim-width is defined within the same framework as mim-width with only a slight
change to its cut function, simply requiring that a special induced matching across a cut cannot contain edges between any pair of vertices on the same side of the cut. This exploits the fact that a cut function for branch decompositions over the vertex set of a graph need not be a parameter of the bipartite graph on edges crossing the cut. The cuts allowed by sim-width are more complex than for mim-width, making sim-width applicable to well-known graph classes having a tree-like structure. We show that chordal graphs and co-comparability graphs have sim-width one and that these decompositions can be found in polynomial time. See Figure~\ref{fig:diagram} for an inclusion diagram of some well-known graph classes illustrating these results.
\begin{figure}[tbh]
\centerline{\includegraphics[scale=0.65]{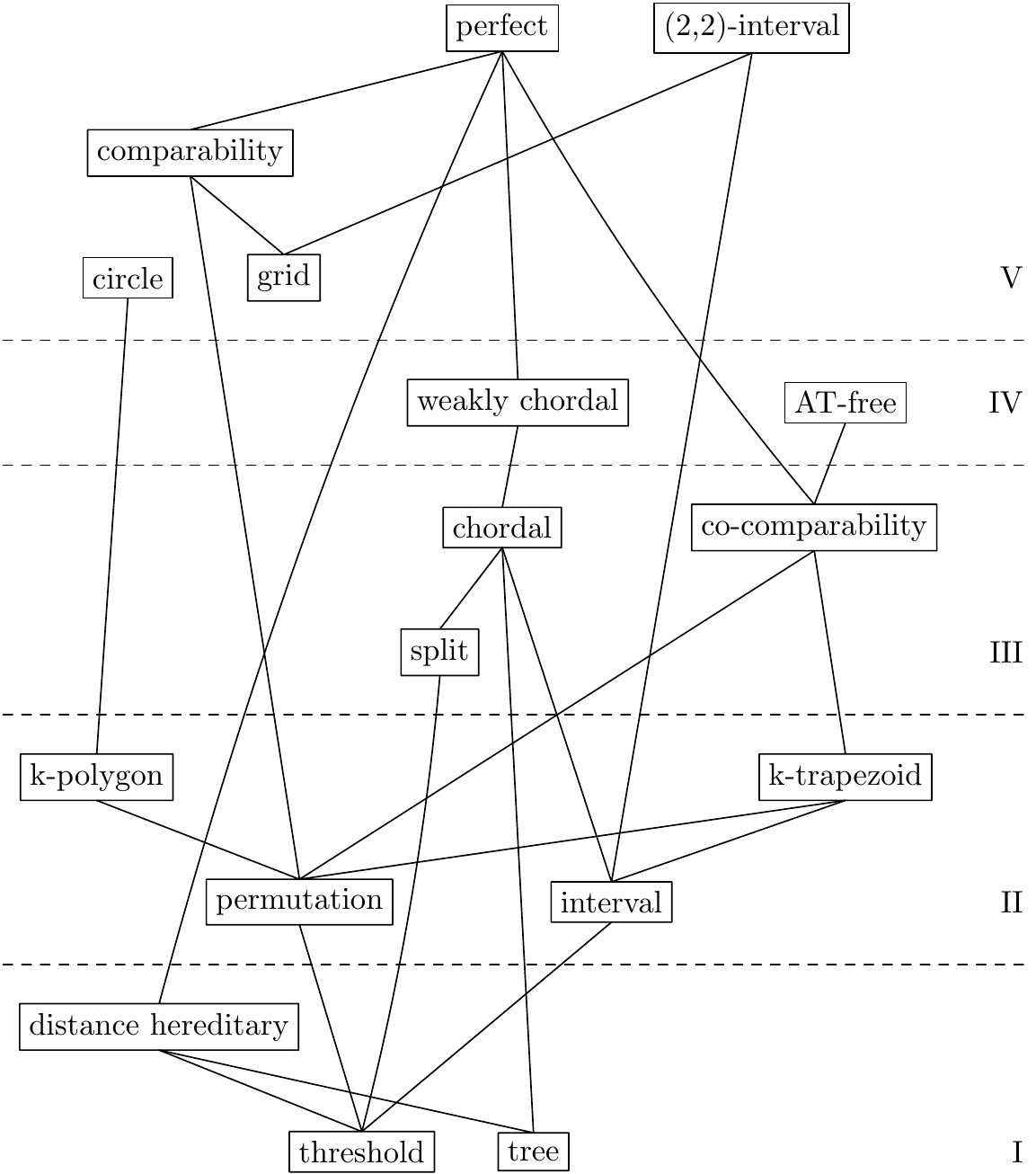}}
\caption{Inclusion diagram of some well-known graph classes.
(I) Classes where clique-width and rank-width are constant.
(II) Classes where mim-width is constant.
(III) Classes where sim-width is constant.
(IV) Classes where it is unknown if sim-width is constant.
(V) Classes where sim-width is unbounded.}
\label{fig:diagram}
\end{figure}
Since LC-VSVP problems like \textsc{Minimum Dominating Set} are NP-complete on chordal graphs we cannot expect XP algorithms parameterized by sim-width. However, for graphs of sim-width $w$, when excluding subgraphs $K_t \mat K_t$ or $K_t \mat S_t$, either as induced subgraphs or induced minors, we get graphs of bounded mim-width. 
The induced minor relation is natural since graphs of bounded sim-width are closed under induced minors, which might be of independent interest when taking the structural point of view. 
We show that by the alternate parametrization $w+t$ the LC-VSVP problems are  solvable in XP time on graphs excluding $K_t \mat K_t$ or $K_t \mat S_t$, when given with a decomposition of sim-width $w$.
The class of circle graphs is one of the classes listed in~\cite{BelmonteV2013} where 
either graphs in the class do not have constant mim-width, or it is NP-complete to find such a decomposition.
Using a technique recently introduced by Mengel~\cite{Mengel} we prove that sim-width of circle graphs is unbounded, which implies that mim-width of circle graphs is also unbounded.

Let us mention some more related work. Golovach et al~\cite{Golovach2015} applied linear mim-width in the context of enumeration algorithms, showing that 
all minimal dominating sets of graphs of bounded linear mim-width can be enumerated with polynomial delay using polynomial space.
The graph parameter mim-width has also been used for knowledge compilation and model counting in the field of satisfiability 
\cite{DBLP:conf/sat/BovaCMS15, DBLP:journals/jair/SaetherTV15, GPST}.
The LC-VSVP problems include the class of domination-type problems known as \textsc{$(\sigma, \rho)$-Domination} problems, whose intractability on chordal graphs is well known~\cite{BoothJ1982}.
For two subsets of natural numbers $\sigma, \rho$ a set $S$ of vertices is called $(\sigma, \rho)$-dominating if for every vertex $v \in S$, $|S \cap N(v)| \in \sigma$, and for every $v \not \in S$, $|S \cap N(v)| \in \rho$. Golovach et al \cite{GK-WG2007} showed that for fixed $\sigma, \rho$ the problem of deciding if a given chordal graph has a $(\sigma, \rho)$-dominating set, is NP-complete whenever there exists at least one chordal graph containing more than one $(\sigma, \rho)$-dominating set, as this graph can then be used as a gadget in a reduction. 
Golovach, Kratochv{\'{\i}}l and Such{\'{y}} \cite{GKS-DAM2012} extended these results to the parameterized setting, showing that 
the existence of a $(\sigma, \rho)$-dominating set of size $k$, and at most $k$, are W[1]-complete problems when parameterized by $k$ for any pair of finite sets $\sigma$ and $\rho$.
In contrast, combining our bounds on mim-width and algorithms of  Bui-Xuan, Telle, and Vatshelle~\cite{BinhminhJM2013} we obtain the following.
\begin{theorem}\label{thm:mainmimw2}
Let $t\ge 2$ be an integer. Given an $n$-vertex $(K_t \mat S_t)$-free chordal graph or an $n$-vertex $(K_t \mat K_t)$-free co-comparability graph,
every fixed LC-VSVP problem can be solved in time $n^{O(t)}$. 
\end{theorem}
As a specific example, we show that \textsc{Minimum Dominating Set} can be solved in time $O(n^{3t+4})$ and \textsc{$q$-Coloring} can be solved in time $O(qn^{3qt+4})$. 
We note that given an $n$-vertex graph, one can test in time $\mathcal{O}(n^{2t})$ whether it contains an induced subgraph isomorphic to $K_t\mat S_t$ or not.
Therefore, a membership testing for $(K_t\mat S_t)$-free chordal graphs and the algorithm in Theorem~\ref{thm:mainmimw2} can be applied in time $n^{\mathcal{O}(t)}$.
The same argument holds for $(K_t\mat K_t)$-free co-comparability graphs.

The remainder of this paper is organized as follows. 
Section~\ref{sec:prelim} contains all the necessary notions required for our results.
In Section~\ref{sec:specialclass}, we prove that chordal graphs have sim-width at most $1$ and mim-width at most $t-1$ if they are $(K_t\mat S_t)$-free. 
Similarly we show that co-comparability graphs have linear sim-width at most $1$ and linear mim-width at most $t-1$ if they are $(K_t\mat K_t)$-free.
We provide polynomial-time algorithms to find such decompositions, and discuss their algorithmic consequences for LC-VSVP problems.
We also show that chordal graphs and co-comparability graphs have unbounded mim-width.
In Section \ref{sec:excludingtmatching}, we bound the mim-width of graphs of sim-width $w$ that do not contain $K_t\mat K_t$ and $K_t\mat S_t$ as induced subgraphs or induced minors, and give algorithmic consequences.
We also show that graphs of bounded sim-width are closed under induced minors. 
Lastly, we show in Section~\ref{sec:circlegraph} that circle graphs have unbounded sim-width.
We finish with some research questions related to sim-width in Section~\ref{sec:conclusion}.

\section{Preliminaries}\label{sec:prelim}

We denote the vertex set and edge set of a graph $G$ by $V(G)$ and $E(G)$, respectively.
Let $G$ be a graph. For a vertex $v$ of $G$, we denote by $N_G(v)$ the set of neighbors of $v$ in $G$.
For $v\in V(G)$ and $X\subseteq V(G)$, we denote by $G-v$ the graph obtained from $G$ by removing $v$, and 
denote by $G-X$ the graph obtained from $G$ by removing all vertices in $X$.
For $e\in E(G)$, we denote by $G-e$ the graph obtained from $G$ by removing $e$, and
denote by $G/e$ the graph obtained from $G$ by contracting $e$.
For a vertex $v$ of $G$ with exactly two neighbors $v_1$ and $v_2$ that are non-adjacent, 
the operation of removing $v$ and adding the edge $v_1v_2$ is called \emph{smoothing} the vertex $v$.
For $X\subseteq V(G)$, we denote by $G[X]$ the subgraph of $G$ induced by $X$.
A \emph{clique} in $G$ is a set of vertices of $G$ that are pairwise adjacent, and 
an \emph{independent set} in $G$ is a set of vertices that are pairwise non-adjacent.
A set of edges $\{v_1w_1, v_2w_2, \ldots, v_mw_m\}$ of $G$ is called an \emph{induced matching} in $G$ 
if 
there are no other edges in $G[\{v_1, \ldots, v_m, w_1, \ldots, w_m\}]$.
A matrix $M$ is called the \emph{adjacency matrix} of $G$ if 
the rows and columns of $M$ are indexed by $V(G)$,
and for $v, w\in V(G)$, 
$M[v,w]=1$ if $v$ and $w$ are adjacent in $G$, and $M[v,w]=0$ otherwise.

A pair of vertex subsets $(A,B)$ of a graph $G$ is called a \emph{vertex bipartition} if $A\cap B=\emptyset$ and $A\cup B=V(G)$.
For a vertex bipartition $(A, B)$ of a graph $G$,
we denote by $G[A, B]$ the bipartite graph on the bipartition $(A, B)$ where
for $a\in A$ and $b\in B$, $a$ and $b$ are adjacent in $G[A, B]$ if and only if they are adjacent in $G$.
For a vertex bipartition $(A,B)$ of $G$ and an induced matching $\{v_1w_1, v_2w_2, \ldots, v_mw_m\}$ in $G$ where $v_1, \ldots, v_m\in A$ and $w_1, \ldots, w_m\in B$, 
we say that it is an induced matching in $G$ between $A$ and $B$.

We denote by $\mathbb{N}$ the set of all non-negative integers, and let $\mathbb{N}^+:=\mathbb{N}\setminus \{0\}$.

\subsection{Graph classes}\label{subsec:graphclass}

A tree is called \emph{subcubic} if every internal node has exactly $3$ neighbors.
A tree $T$ is called a \emph{caterpillar} if it contains a path $P$ such that every vertex in $V(T)\setminus V(P)$ has a neighbor in $P$.
 The complete graph on $n$ vertices is denoted by $K_n$.

A graph is \emph{chordal} if it contains no induced subgraph isomorphic to a cycle of length 4 or more. 
A graph is a \emph{split} graph if it can be partitioned into two vertex sets $C$ and $I$ where $C$ is a clique and $I$ is an independent set.
A graph is an \emph{interval graph} if it is the intersection graph of a family of intervals on the real line.
Every split graph and interval graph is chordal.
An ordering $v_1, \ldots, v_n$ of the vertex set of a graph $G$ is called a \emph{co-comparability ordering} 
if for every integers $i,j,k$ with $1\le i<j<k\le n$, 
\begin{itemize}
\item if $v_i$ is adjacent to $v_k$, 
 then $v_j$ is adjacent to $v_i$ or $v_k$. 
 \end{itemize}
 This condition is equivalent to the following: for every integers $i,j,k$ with $1\le i<j<k\le n$, 
 $v_j$ has a neighbor in every path from $v_i$ to $v_k$ avoiding $v_j$.
 A graph is a \emph{co-comparability graph} if it admits a co-comparability ordering.
 Every co-comparability graph is the complement of some comparability graph, where comparability graphs 
 are graphs that can be obtained from some partial order by connecting pairs of elements that are comparable to each other.
A graph is a \emph{permutation graph} if it is the intersection graph of line segments whose endpoints lie on two parallel lines.
Permutation graphs are co-comparability graphs~\cite{Dushnik1941}.
A graph is a \emph{circle graph} if it is the intersection graph of chords in a circle.

For positive integer $n$, 
let $K_n\mat K_n$ be the graph on $\{v_1^1,\ldots,v_n^1, v_1^2,\ldots, v_n^2\}$
such that  for all $i,j\in \{1,\ldots,n\}$, 
\begin{itemize}
\item $\{v_1^1,\ldots,v_n^1\}$ and $\{v_1^2,\ldots,v_n^2\}$ are cliques,
\item $v^1_i$ is adjacent to $v^2_j$ if and only if $i=j$, 
\end{itemize}
and let $K_n\mat S_n$ be the graph on $\{v_1^1,\ldots,v_n^1, v_1^2,\ldots, v_n^2\}$
such that  for all $i,j\in \{1,\ldots,n\}$, 
\begin{itemize}
\item $\{v_1^1,\ldots,v_n^1\}$ is a clique, $\{v_1^2,\ldots,v_n^2\}$ is an independent set,
\item $v^1_i$ is adjacent to $v^2_j$ if and only if $i=j$.
\end{itemize}
Since $K_2\mat K_2$ is an induced cycle of length $4$, chordal graphs do not contain $K_2\mat K_2$ as an induced subgraph.
We observe that $K_3\mat S_3$ is not a co-comparability graph. 
\begin{lemma}\label{lem:obscocom}
The graph $K_3\mat S_3$ is not a co-comparability graph.
\end{lemma}
\begin{proof}
Let $G$ be a graph on $\{v_1, v_2, v_3\}\cup \{w_1, w_2, w_3\}$ such that
\begin{itemize}
\item $\{v_1, v_2, v_3\}$ is a clique, $\{w_1, w_2, w_3\}$ is an independent set, and
\item $v_i$ is adjacent to $w_j$ if and only if $i=j$.
\end{itemize}
It is clear that $G$ is isomorphic to $K_3\mat S_3$.
Suppose $G$ admits a co-comparability ordering.
By relabeling if necessary, we may assume $w_1, w_2, w_3$ appear in the co-comparability ordering in that order.
However, there is a path $w_1v_1v_3w_3$ from $w_1$ to $w_3$ avoiding $v_2$ and $w_2$, and thus, 
it contradicts to the assumption.
We conclude that $K_3\mat S_3$ is not a co-comparability graph.
\end{proof}

\subsection{Graph relations}\label{subsec:graphrelation}

Let $G$ be a graph.
A graph $H$ is a \emph{subgraph} of $G$ if $H$ can be obtained from $G$ by removing some vertices and edges.
A graph $H$ is an \emph{induced subgraph} of $G$ if $H=G[X]$ for some $X\subseteq V(G)$.
A graph $H$ is an \emph{induced minor} of $G$ if $H$ can be obtained from $G$ by a sequence of removing vertices and contracting edges.
A graph $H$ is a \emph{minor} of $G$ if $H$ can be obtained from $G$ by a sequence of removing vertices, removing edges, and contracting edges.
For a graph $H$, we say a graph is \emph{$H$-free} if it contains no induced subgraph isomorphic to $H$.

A \emph{minor model} of a graph $H$ in $G$ is a function $\eta$ with the domain $V(H)\cup E(H)$, where
\begin{itemize}
\item for every $v\in V(H)$, $\eta(v)$ is a non-empty connected subgraph of $G$, all pairwise vertex-disjoint
\item for every edge $e$ of $H$, $\eta(e)$ is an edge of $G$, all distinct
\item if $e\in E(H)$ and $v\in V(H)$ then $\eta(e)\notin E(\eta(v))$,
\item for every edge $e=uv$ of $H$,  $\eta(e)$ has one end in $V(\eta(u))$ and
     the other in $V(\eta(v))$.
\end{itemize}
It is well known that $H$ is a minor of $G$ if and only if there is a minor model of $H$ in $G$.
A minor model $\eta$ of a graph $H$ in $G$ is an \emph{induced minor model} if 
for every distinct vertices $v_1$ and $v_2$ in $H$ that are non-adjacent, 
there are no edges between $V(\eta(v_1))$ and $V(\eta(v_2))$.
A graph $H$ is an induced minor of $G$ if and only if there is an induced minor model of $H$ in $G$.

\subsection{Width parameters}\label{subsec:widthparameter}

For sets $A$ and $B$, a function $f:2^A\rightarrow B$ is \emph{symmetric} if for every $Z\subseteq A$, $f(Z)=f(A\setminus Z)$.

For a graph $G$ and a vertex set $A \subseteq V(G)$, we define functions $\cutrk_{G}$, $\mimval_{G}$, and $\simval_{G}$ from $2^{V(G)}$ to $\mathbb{N}$ such that 
\begin{itemize}
\item $\cutrk_{G}(A)$ is the rank of the matrix $M[A, V(G)\setminus A]$ where $M$ is the adjacency matrix of $G$ and the rank is computed over the binary field,
\item $\mimval_{G}(A)$ is the maximum size of an induced matching of $G[A, V(G)\setminus A]$,
\item $\simval_G(A)$ is the maximum size of an induced matching between $A$ and $V(G)\setminus A$ in $G$. 
\end{itemize}
For a graph $G$, a pair $(T,L)$ of a subcubic tree $T$ and a bijection $L$ from $V(G)$ to the set of leaves of $T$ is called a \emph{branch-decomposition}.
For each edge $e$ of $T$, 
let $T^e_1$ and $T^e_2$ be the two connected components of $T-e$, and 
let $(A^e_1, A^e_2)$ be the vertex bipartition of $G$ such that for each $i\in \{1,2\}$, 
$A^e_i$ is the set of all vertices in $G$ mapped to leaves contained in $T^e_i$ by $L$. 
We call $(A^e_1, A^e_2)$ the vertex bipartition of $G$ associated with $e$. 
For a branch-decomposition $(T,L)$ of a graph $G$ and an edge $e$ in $T$ and a symmetric function $f:2^{V(G)}\rightarrow \mathbb{N}$, 
the \emph{$f$-width} of $e$ is define as $f(A^e_1)$ where $(A^e_1, A^e_2)$ is the vertex bipartition associated with $e$.
The \emph{$f$-width} of $(T,L)$ is the maximum $f$-width over all edges in $T$, and
the \emph{$f$-width} of $G$ is the minimum $f$-width over all branch-decompositions of $G$.
If $\abs{V(G)}\le 1$, then $G$ does not admit a branch-decomposition, and the $f$-width of $G$ is defined to be $0$.

The \emph{rank-width} of a graph $G$, denoted by $\rw(G)$, is the $\cutrk_G$-width of $G$, and
the \emph{mim-width} of a graph $G$, denoted by $\mimw(G)$, is the $\mimval_G$-width of $G$, and
the \emph{sim-width} of a graph $G$, denoted by $\simw(G)$, is the $\simval_G$-width of $G$.
For convenience, the $f$-width of a branch-decomposition is also called a rank-width, mim-width, or sim-width depending on the function $f$.

If $T$ is a subcubic caterpillar tree, then a branch-decomposition $(T,L)$ is called a \emph{linear branch-decomposition}.
The \emph{linear $f$-width} of $G$ is the minimum $f$-width over all linear branch-decompositions of $G$.
The \emph{linear mim-width} of a graph $G$, denoted by $\lmimw(G)$, is the linear $\mimval_G$-width of $G$, and
the \emph{linear sim-width} of a graph $G$, denoted by $\lsimw(G)$, is the linear $\simval_G$-width of $G$.

By definitions we have the following.

\begin{lemma}\label{lem:inequality}
For a graph $G$, we have $\simw(G)\le \mimw(G)\le \rw(G)$.
\end{lemma}

We frequently use the following fact.
\begin{lemma}\label{lem:balancedpartition}
Let $G$ be a graph, let $w$ be a positive integer, and let $f:2^{V(G)}\rightarrow \mathbb{N}$ be a symmetric function.
If $G$ has $f$-width at most $w$, then 
$G$ admits a vertex bipartition $(A_1, A_2)$ where $f(A_1)\le w$ and for each $i\in \{1, 2\}$, $\frac{\abs{V(G)}}{3}<\abs{A_i}\le \frac{2\abs{V(G)}}{3}$.
\end{lemma}
\begin{proof}
Let $(T,L)$ be a branch-decomposition of $G$ of $f$-width at most $w$.
We subdivide an edge of $T$, and regard the new vertex as a root node.
For each node $t\in V(T)$, let $\mu(t)$ be the number of leaves of $T$ that are descendants of $t$.
Now, we choose a node $t$ that is farthest from the root node such that 
$\mu(t)> \frac{\abs{V(G)}}{3}$.
By the choice of $t$, for each child $t'$ of $t$, $\mu(t')\le \frac{\abs{V(G)}}{3}$.
Therefore, $\frac{\abs{V(G)}}{3}<\mu(t)\le \frac{2\abs{V(G)}}{3}$.
Let $e$ be the edge connecting the node $t$ and its parent.
By the construction, the vertex bipartition associated with $e$ is a required bipartition.
\end{proof}

We also use tree-decompositions in Section~\ref{sec:specialclass}.
A \emph{tree-decomposition} of a graph $G$ is a pair $(T, \cB=\{B_t\}_{t\in V(T)})$ such that  
\begin{enumerate}[(1)]
\item $\bigcup_{t\in V(T)}B_t=V(G)$,
\item for every edge in $G$, there exists $B_t$ containing both end vertices, 
\item for $t_1, t_2, t_3\in V(T)$, $B_{t_1}\cap B_{t_3}\subseteq B_{t_2}$ whenever $t_2$ is on the path from $t_1$ to $t_3$. 
\end{enumerate}
Each vertex subset $B_t$ is called a \emph{bag} of the tree-decomposition. The \emph{width} of a tree-decomposition is $w-1$ where $w$ is the maximum size of a bag in the tree-decomposition, and the \emph{tree-width} of a graph is the minimum width over all tree-decompositions of the graph. It is well known that 
a graph has tree-width at most $w$ if and only if it is a subgraph of a chordal graph with maximum clique size at most $w+1$; see for instance \cite{Bodlaender98}. Furthermore, chordal graphs admit a tree-decomposition where each bag induces a maximal clique of the graph. We will use this fact in Section~\ref{sec:specialclass}.

\subsection{LC-VSVP problems}\label{subsec:lcvsvp}
Telle and Proskurowski~\cite{TelleP1997} classified a class of problems called \emph{locally checkable vertex subset and vertex partitioning problems}, which is a subclass of $\operatorname{MSO}_1$ problems. These problems generalize problems like \textsc{Maximum Independent Set}, \textsc{Minimum Dominating Set}, \textsc{$q$-Coloring} etc.

Let $\sigma$ and $\rho$ be finite or co-finite subsets of $\mathbb{N}$. For a graph $G$ and $S\subseteq V(G)$, 
we call $S$ a \emph{$(\sigma, \rho)$-dominating set} of $G$ if 
\begin{itemize}
\item for every $v\in S$, $\abs{N_G(v)\cap S}\in \sigma$, and
\item for every $v\in V(G)\setminus S$, $\abs{N_G(v)\cap S}\in \rho$.
\end{itemize}
For instance, a $(0, \mathbb{N})$-set is an independent set as there are no edges inside of the set, and we do not care about the adjacency between $S$ and $V(G)\setminus S$.
Another example is that an $(\mathbb{N}, \mathbb{N}^+)$-set is a dominating set as we require that for each vertex in $V(G)\setminus S$, it has at least one neighbor in $S$.
See \cite[Table 4.1]{TelleP1997} for more examples.
The \textsc{Min-(or Max-)$(\sigma, \rho)$-domination} problem is a problem to find a minimum (or maximum) $(\sigma,\rho)$-dominating set in an input graph $G$, and possibly on vertex-weighted graphs. These problems also called as  \emph{locally checkable vertex subset} problems.

For a positive integer $q$, a $(q\times q)$-matrix $D_q$ is called a \emph{degree constraint} matrix if each element is either a finite or co-finite subset of $\mathbb{N}$.
A partition $\{V_1, V_2, \ldots, V_q\}$ of the vertex set of a graph $G$ is called a \emph{$D_q$-partition} if
\begin{itemize}
\item for every $i,j\in \{1, \ldots, q\}$ and $v\in V_i$, $\abs{N_G(v)\cap V_j}\in D_q[i,j]$.
\end{itemize}
For instance, if we take a matrix $D_q$ where all diagonal entries are $0$, and all other entries are $\mathbb{N}$, 
then a $D_q$-partition is a partition into $q$ independent sets, which corresponds to a $q$-coloring of the graph.
The \textsc{$D_q$-partitioning} problem is a problem deciding if an input graph admits a $D_q$-partition or not.
These problems are also called as \emph{locally checkable vertex partitioning} problems.

All these problems will be called \emph{locally checkable vertex subset and vertex partitioning problems}, shortly LC-VSVP problems.
As shown in \cite{BinhminhJM2013} the runtime solving an LC-VSVP problem by dynamic programming relates to the finite or co-finite subsets of $\mathbb{N}$ used in its definition. 
The following function $d$ is central.
\begin{enumerate}
\item Let $d(\mathbb{N})=0$.
\item For every finite or co-finite set $\mu\subseteq \mathbb{N}$, let $d(\mu)=1+ \min( \max \{x\in \mathbb{N}:x\in \mu\}, \max \{x\in \mathbb{N}:x\notin \mu\})$. 
\end{enumerate}
For example, for \textsc{Minimum Dominating Set} and \textsc{$q$-Coloring} problems we plug in $d=1$ because $\max( d(\mathbb{N}), d(\mathbb{N}^+))=1$ and  $\max( d(0), d(\mathbb{N}))=1$.

Belmonte and Vatshelle~\cite{BelmonteV2013}  proved the following application of mim-width.

\begin{theorem}[Belmonte and Vatshelle~\cite{BelmonteV2013} and Bui-Xuan, Telle, and Vatshelle~\cite{BinhminhJM2013}]\label{thm:mimalgvs}
\begin{enumerate}[(1)]
\item Given an $n$-vertex graph and its branch-decomposition $(T,L)$ of mim-width $w$ and $\sigma, \rho \subseteq \mathbb{N}$,
one can solve \textsc{Min-(or Max-)$(\sigma, \rho)$-domination}  in time $\mathcal{O}(n^{3dw+4})$ where $d=\max( d(\sigma), d(\rho))$.
\item Given an $n$-vertex graph and its branch-decomposition $(T,L)$ of mim-width $w$ and a $(q\times q)$-matrix $D_q$,
one can solve \textsc{$D_q$-partitioning} in time $\mathcal{O}(qn^{3dwq+4})$ where $d=\max_{i,j} d(D_q[i,j])$.
\end{enumerate}
\end{theorem}

\section{Mim-width of chordal graphs and co-comparability graphs}\label{sec:specialclass}

In this section, we show that chordal graphs admit a branch-decomposition $(T,L)$ such that
\begin{enumerate}
\item $(T, L)$ has sim-width at most $1$, and
\item $(T,L)$ has mim-width at most $t-1$ if the given graph is $(K_t\mat S_t)$-free for some $t\ge 2$.
\end{enumerate}
We combine the second statement with Theorem~\ref{thm:mimalgvs}, and show that LC-VSVP problems can be efficiently solved for $(K_t\mat S_t)$-free chordal graphs (Corollary~\ref{cor:mainchordal}).
In the same context, we show that co-comparability graphs admit a linear branch-decomposition such that
\begin{enumerate}
\item $(T,L)$ has linear sim-width at most $1$, and
\item $(T,L)$ has linear mim-width at most $t-1$ if the given graph is $(K_t\mat K_t)$-free for some $t\ge 2$.
\end{enumerate}
We combine the second statement with Theorem~\ref{thm:mimalgvs}, and show that LC-VSVP problems can be efficiently solved for $(K_t\mat K_t)$-free co-comparability graphs (Corollary~\ref{cor:maincocomparability}).
We also prove that chordal graphs and co-comparability graphs have unbounded mim-width in general.

\subsection{Mim-width of chordal graphs}\label{subsec:chordal}

We prove the following.
\begin{proposition}\label{prop:chordal}
Given a chordal graph $G$, one can in time $\mathcal{O}(\abs{V(G)}+\abs{E(G)})$ output a branch-decomposition $(T,L)$ with the following property:
\begin{itemize}
\item $(T,L)$ has sim-width at most $1$,
\item $(T,L)$ has mim-width at most $t-1$ if $G$ is $(K_t\mat S_t)$-free for some integer $t\ge 2$.
\end{itemize}
\end{proposition}
To show Proposition~\ref{prop:chordal}, we use the fact that chordal graphs admit a tree-decomposition whose bags are maximal cliques.
Such a tree-decomposition can be easily transformed from a \emph{clique-tree} of a chordal graph.  
A \emph{clique-tree} of a chordal graph $G$ is a pair $(T, \{C_t\}_{t\in V(T)})$, where 
\begin{enumerate}[(1)]
\item $T$ is a tree, 
\item each $C_t$ is a maximal clique of $G$, and 
\item for each $v\in V(G)$, the vertex set $\{t\in V(T):v\in C_t\}$ induces a subtree of $T$.
\end{enumerate}
Gavril~\cite{Gavril1974} showed that chordal graphs admit a clique-tree, and
it is known that given a chordal graph $G$, its clique-tree can be constructed in time $\mathcal{O}(\abs{V(G)}+\abs{E(G)})$~\cite{BlairP1993, HsuM1991}.

\begin{proof}[Proof of Proposition~\ref{prop:chordal}]
Let $G$ be a chordal graph.
We compute a tree-decomposition $(F,\cB=\{B_t\}_{t\in V(F)})$ of $G$ whose 
bags induce maximal cliques of $G$ in time $\mathcal{O}(\abs{V(G)}+\abs{E(G)})$. 
Let us choose a root node $r$ of $F$. 

\begin{figure}[t]
\centerline{\includegraphics[scale=0.4]{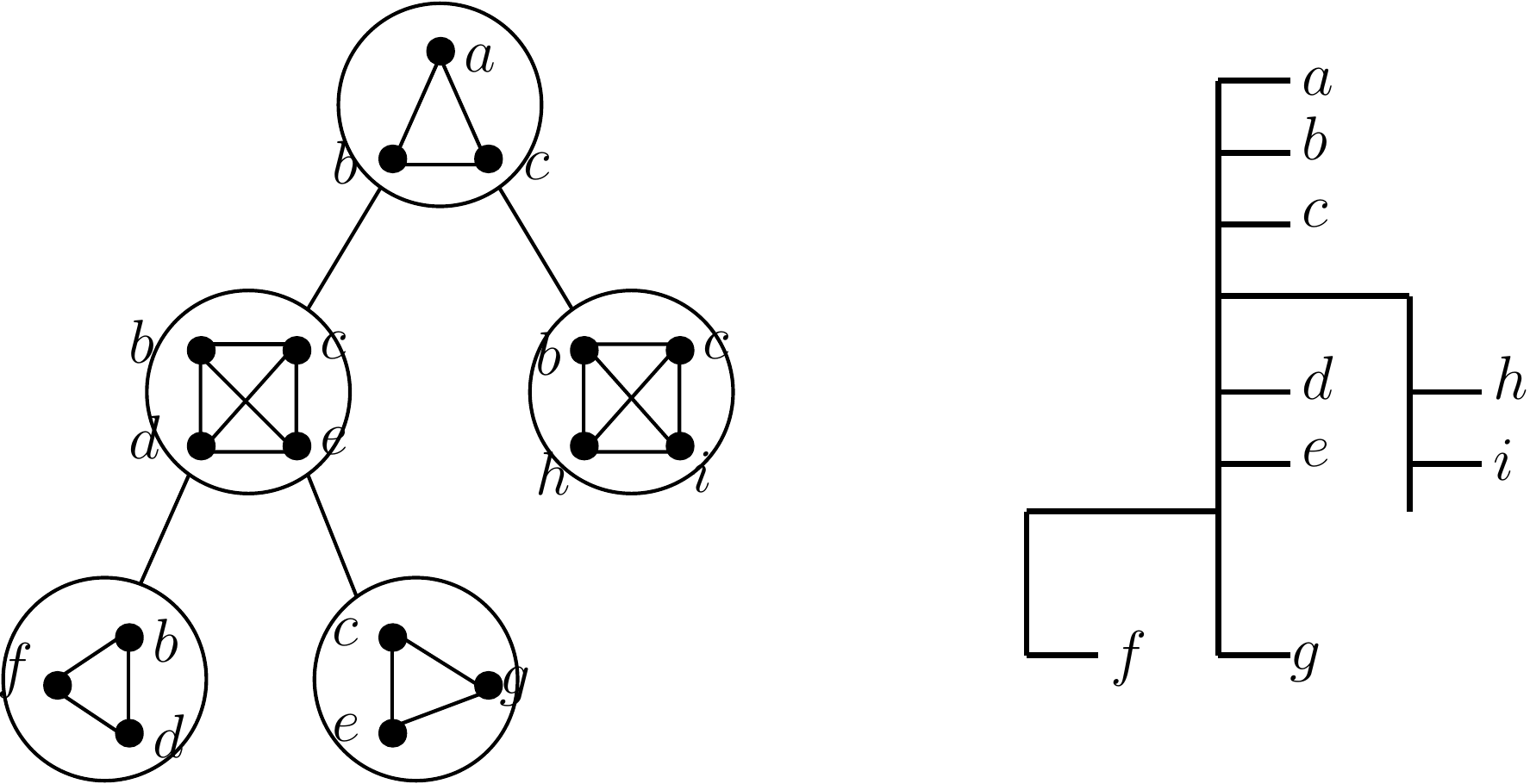}}
\caption{Constructing a branch-decomposition $(T,L)$ of a chordal graph $G$ of sim-width at most $1$ from its tree-decomposition. }
\label{fig:rootedtree}
\end{figure}

We construct a tree $(T, L)$ from $F$ as follows.
\begin{enumerate}
\item We attach a leaf $r'$ to the root node $r$ of $F$ and regard it as the parent of $r$ and let $B_{r'}:=\emptyset$. 
\item For every $t\in V(F)$ with its parent $t'$, we subdivide the edge $tt'$ into a path $tv^t_1 \cdots v^t_{\abs{B_t\setminus B_{t'}}}t'$, 
and for each $j\in \{1, \ldots, \abs{B_t\setminus B_{t'}}\}$, 
we attach a leaf $z^t_j$ to $v^t_j$ and assign the vertices of $B_t\setminus B_{t'}$ to
$L(z^t_1), \ldots, L(z^t_{\abs{B_t\setminus B_{t'}}})$ in any order.  Then remove $r'$.
\item We transform the resulting tree into a tree of maximum degree at most $3$. For every $t\in V(F)$, we do the following. Let $t_1, \ldots, t_m$ be the children of $t$ in $F$. 
We remove $t$ and introduce a path $w^t_1w^t_2 \cdots w^t_{m}$. If $t$ is a leaf, then we just remove it.
We add an edge $w^t_1v^t_1$, 
and for each $i\in \{1, \ldots, m\}$, add an edge $w^t_i v^{t_i}_{\abs{B_{t_i}\setminus B_{t}}}$.
\item Let $T'$ be the resulting tree, and we obtain a tree $T$ from $T'$ by smoothing all nodes of degree $2$.
Let $(T,L)$ be the resulting branch-decomposition. See Figure~\ref{fig:rootedtree} for an illustration of the construction. 
\end{enumerate}
We can construct $(T,L)$ in linear time as the number of nodes in $T$ is $\mathcal{O}(\abs{V(G)})$.
We consider $T$ as a rooted tree with the root $z^r_{B_r\setminus B_{r'}}$.

We claim that $(T,L)$ has sim-width at most $1$. We prove a stronger result that for every edge $e$ of $T$ with a vertex bipartition $(A,B)$ associated with $e$, 
either $N_G(A)\cap B$ or $N_G(B)\cap A$ is a clique.

\begin{claim}\label{claim:clique}
Let $e$ be an edge of $T$ and let $(A,B)$ be the vertex bipartition of $G$ associated with $e$.
Then either $N_G(A)\cap B$ or $N_G(B)\cap A$ is a clique.
\end{claim}
\begin{clproof}
For convenience, we prove for $T'$, which is the tree before smoothing.
We may assume that both end nodes of $e$ are internal nodes of $T'$, otherwise, it is trivial.
There are four types of $e$:
\begin{enumerate}
\item $e=v^t_iv^t_{i+1}$ for some $t\in V(F)$, its parent $t'$, and $i\in \{1, \ldots, \abs{B_t\setminus B_{t'}}-1\}$.
\item $e=w^t_1v^t_1$ for some $t\in V(F)$.
\item $e=v^{t_i}_{\abs{B_{t_i}\setminus B_t}}w^t_i$ for some $t\in V(F)$ and its child $t_i$.
\item $e=w^t_iw^t_{i+1}$ for some $t\in V(F)$ and its child $t_i$.
\end{enumerate}
Suppose $e=v^t_iv^t_{i+1}$ for some $t$ and $i$, and let $t'$ be the parent of $t$. 
We may assume that $A$ is the set of all vertices assigned to the descendants of $v^t_i$.
Note that $B_t\cap B_{t'}$ separates $A$ and $B\setminus B_t$  in $G$.
Therefore, for each $v\in A$, $N_G(v)\cap B$ is contained in $B_t$, and 
$N_G(A)\cap B$ is a clique.
We can similarly prove for Cases 2 and 3.

Now, let $t_1, \ldots, t_m$ be the children of $t$, and let $w^t_1, \ldots, w^t_m$ be the vertices that were replaced from $t$ in the algorithm. 
We assume $e=w^t_iw^t_{i+1}$ for some $i\in \{1, \ldots, m-1\}$. For each $j\in \{1, \ldots, m\}$, let $L_j$ be the set of all vertices assigned to the descendants of $v^{t_j}_{\abs{B_{t_j}\setminus B_t}}$. Without loss of generality, we may assume that $B_t\subseteq B$.
We can observe that $A=\bigcup_{i+1\le j\le m} L_j$.

Let $j_1$ and $j_2$ be integers such that $1\le j_1\le i<j_2\le m$. Note that $B_t$ separates $L_{j_1}$ and $L_{j_2}$ in $G$.
Therefore, for every $v\in A$, $N_G(v)\cap B$ is contained in $B_t$, and $N_G(A)\cap B$ is a clique, as required.
\end{clproof}

We prove that when $G$ is $(K_t\mat S_t)$-free for some $t\ge 2$, $(T,L)$ has mim-width at most $t-1$.

\begin{claim}
Let $t\ge 2$ be an integer. If $G$ is $(K_t\mat S_t)$-free, then $(T,L)$ has mim-width at most $t-1$.
\end{claim}
\begin{clproof}
We show that $(T, L)$ has mim-width at most $t-1$. 
Suppose for contradiction that there is an edge $e$ of $T$ with the vertex bipartition $(A,B)$ of $G$ associated with $e$
such that $\mimval_{G}(A)\ge t$. 
We may assume both end nodes of $e$ are internal nodes of $T$.

By Claim~\ref{claim:clique}, one of $N_G(A)\cap B$ and $N_G(B)\cap A$ is a clique. Without loss of generality we assume $N_G(B)\cap A$ is a clique $C$.
Since $\mimval_{G}(A)\ge t$, 
there is an induced matching $\{a_1b_1, \ldots, a_tb_t\}$ in $G[A, B]$ where $a_1, \ldots, a_t\in A$. 
Since $N_G(B)\cap A$ is a clique $C$, there are no edges between vertices in $\{b_1, \ldots, b_t\}$, otherwise, it creates an induced cycle of length $4$.
Thus, we have an induced subgraph isomorphic to $K_t\mat S_t$, which contradicts to our assumption.
We conclude that  $(T,L)$ has mim-width at most $t-1$.
\end{clproof}

This concludes the proposition.
\end{proof}

As a corollary of Proposition~\ref{prop:chordal}, we obtain the following.

\begin{corollary}\label{cor:mainchordal}
Let $t\ge 2$ be an integer.
\begin{enumerate}[(1)]
\item Given an $n$-vertex $(K_t\mat S_t)$-free chordal graph and $\sigma, \rho \subseteq \mathbb{N}$,
one can solve \textsc{Min-(or Max-) $(\sigma, \rho)$-domination}  in time $\mathcal{O}(n^{3d(t-1)+4})$ where $d=\max( d(\sigma), d(\rho))$.
\item Given an $n$-vertex $(K_t\mat S_t)$-free chordal graph and a $(q\times q)$-matrix $D_q$,
one can solve \textsc{$D_q$-partitioning} in time $\mathcal{O}(qn^{3dq(t-1)+4})$ where $d=\max_{i,j} d(D_q[i,j])$.
\end{enumerate}
\end{corollary}

This generalizes the algorithms for interval graphs, as interval graphs are $(K_3\mat S_3)$-free chordal graphs~\cite{LekkerkerkerB1962}.

We now prove the lower bound on the mim-width of general chordal graphs.
We show this for split graphs, which form a subclass of the class of chordal graphs.
The Sauer-Shelah lemma~\cite{Sauer1972, Shelah1972} is essential in the proof.

\begin{theorem}[Sauer-Shelah lemma~\cite{Sauer1972, Shelah1972}]\label{thm:sauershelah}
Let $t$ be a positive integer and let $M$ be an $X \times Y$ $(0,1)$-matrix such that $\abs{Y}\ge 2$ and any two row vectors of $M$ are distinct.
If $\abs{X}\ge \abs{Y}^t$, then there are $X'\subseteq X$, $Y'\subseteq Y$ such that $\abs{X'}=2^t$, $\abs{Y'}=t$, and all possible row vectors of length $t$ appear in $M[X',Y']$.
\end{theorem}

\begin{proposition}\label{prop:lowerbdsplit}
For every large enough $n$, there is a split graph on $n$ vertices having mim-width at least $\frac{\log_2 n}{5\log_2 (\log_2 n)}$.
\end{proposition}
\begin{proof}
Let $m\ge 4$ be an integer and let $n:=m+(2^{m}-1)$.
Let $G$ be a split graph on the vertex bipartition $(C,I)$ where
$C$ is a clique of size $m$, $I$ is an independent set of size $2^m-1$, and all vertices in $I$ have pairwise distinct non-empty neighborhoods on $C$.
We claim that every branch-decomposition of $G$ has mim-width at least  $\frac{\log_2 n}{5\log_2 (\log_2 n)}$.

Let $(T,L)$ be a branch-decomposition of $G$.
By Lemma~\ref{lem:balancedpartition}, there is a vertex bipartition $(A_1, A_2)$ of $G$ associated with $e$ satisfies that 
for each $i\in \{1,2\}$, $\frac{n}{3}<\abs{A_i}\le \frac{2n}{3}$.
Without loss of generality, we may assume that $\abs{A_1\cap C}\ge \abs{A_2\cap C}$, and thus we have $\frac{m}{2}\le \abs{A_1\cap C}\le m$ and $\abs{A_2\cap C}\le \frac{m}{2}$.

\begin{claim}
There are at least $\abs{A_1\cap C}^{\frac{m}{4\log_2 m}}$ vertices in $A_2\cap I$ that have pairwise distinct neighbors in $A_1\cap C$. 
\end{claim}
\begin{clproof}
Since $m\ge 4$, we have $\abs{A_2\cap I}=\abs{A_2}-\abs{A_2\cap C}> \frac{n}{3}-\frac{m}{2}\ge \frac{2^m-\frac{m}{2}-1}{3}\ge \frac{1}{4}2^{m}$.
Since $\abs{A_2\cap C}< \frac{m}{2}$, 
there are at least 
\[\frac{1}{4}\cdot \frac{2^m}{2^{\abs{A_2\cap C}}}\ge \frac{1}{4}2^{\frac{m}{2}}\]
vertices in $A_2\cap I$ that have pairwise distinct neighbors on $A_1\cap C$. 
Thus, we have 
\[\frac{1}{4}2^{\frac{m}{2}} \ge m^{\frac{m}{4\log_2 m}} \ge \abs{A_1\cap C}^{\frac{m}{4\log_2 m}},\]
as required.
\end{clproof}

Let $I'\subseteq A_2\cap I$ be a maximal set of vertices that have pairwise distinct neighbors in $A_1\cap C$.
By the Sauer-Shelah lemma, there is an induced matching of size $\frac{m}{4\log_2 m}$ between $A_1\cap C$ and $I'$.
It implies that $(T, L)$ has mim-width at least $\frac{m}{4\log_2 m}$.
As $(T,L)$ was chosen arbitrary, the mim-width of $G$ is at least $\frac{m}{4\log_2 m}\ge \frac{(\log_2 n) -1}{4\log_2 (\log_2 n)}\ge \frac{\log_2 n}{5\log_2 (\log_2 n)}$.
\end{proof}
 
 \subsection{Mim-width of co-comparability graphs}\label{subsec:cocomparability}
 
We observe similar properties for co-comparability graphs.
 
 \begin{theorem}[McConnell and Spinrad~\cite{RossJ1994}]\label{thm:spinrad}
 Given a co-comparability graph $G$, one can output a co-comparability ordering in time $\mathcal{O}(\abs{V(G)}+\abs{E(G)})$.
 \end{theorem}

\begin{proposition}\label{prop:cocomparability}
Given a co-comparability graph $G$, 
one can in time $\mathcal{O}(\abs{V(G)}+\abs{E(G)})$ output a linear branch-decomposition $(T,L)$ with the following property: 
\begin{itemize}
\item $(T,L)$ has linear sim-width at most $1$,
\item $(T,L)$ has linear mim-width at most $t-1$ if $G$ is $(K_t\mat K_t)$-free for some integer $t\ge 2$.
\end{itemize}
\end{proposition}

\begin{proof}
Let $G$ be a co-comparability graph.
Using Theorem~\ref{thm:spinrad}, we can obtain its co-comparability ordering $v_1, \ldots, v_n$ in time $\mathcal{O}(\abs{V(G)}+\abs{E(G)})$.
From this, we take a linear branch-decomposition $(T,L)$ following the co-comparability ordering.
We claim that for each $i\in \{2, \ldots, n-1\}$,  there is no induced matching of size $2$ between $\{v_1, \ldots, v_i\}$ and $\{v_{i+1}, \ldots, v_n\}$.
Suppose there are $i_1, i_2\in \{1, \ldots, i\}$ and $j_1, j_2\in \{i+1, \ldots, n\}$ such that $\{v_{i_1}v_{j_1}, v_{i_2}v_{j_2}\}$ is an induced matching of $G$.
Without loss of generality we may assume that $i_1<i_2$.
Then we have $i_1<i_2<j_1$, and thus by the definition of the co-comparability ordering, $v_{i_2}$ should be adjacent to one of $v_{i_1}$ and $v_{j_1}$, which contradicts to our assumption.
Therefore, there is no induced matching of size $2$ between $\{v_1, \ldots, v_i\}$ and $\{v_{i+1}, \ldots, v_n\}$. 
It implies that $(T,L)$ has sim-width at most $1$.

Now, suppose that $G$ is $(K_t\mat K_t)$-free for $t\ge 2$. 
We prove that for every $i\in \{1, \ldots, n\}$, there are no induced matchings of size $t$ in $G[\{v_1, \ldots, v_i\}, \{v_{i+1}, \ldots, v_n\}]$. 
For contradiction, suppose there is an induced matching $\{v_{i_1}v_{j_1}, \ldots, v_{i_t}v_{j_t}\}$ in $G[\{v_1, \ldots, v_i\}, \{v_{i+1}, \ldots, v_n\}]$ where $i_1, \ldots, i_t\in \{1, \ldots, i\}$. 
We claim that $\{v_{i_1}, \ldots, v_{i_t}\}$ and $\{v_{j_1}, \ldots, v_{j_t}\}$ are cliques.

\begin{claim}\label{claim:twocliques}
 $\{v_{i_1}, \ldots, v_{i_t}\}$ and $\{v_{j_1}, \ldots, v_{j_t}\}$ are cliques.
\end{claim}
\begin{clproof}
Let $x,y\in \{1, \ldots, t\}$.
Assume that $i_x<i_y$. Then $i_x<i_y<j_x$, and therefore, $v_{i_y}$ is adjacent to either $v_{i_x}$ or $v_{j_x}$.
Since $v_{i_y}$ is not adjacent to $v_{j_x}$, $v_{i_y}$ is adjacent to $v_{i_x}$. In case when $i_y<i_x$, we also have that $v_{i_y}$ is adjacent to $v_{i_x}$ by the same reason.
It implies that  $\{v_{i_1}, \ldots, v_{i_t}\}$ is a clique.
By the symmetric argument, we have that $\{v_{j_1}, \ldots, v_{j_t}\}$ is a clique.
\end{clproof}

By Claim~\ref{claim:twocliques}, $\{v_{i_1}, \ldots, v_{i_t}\}$ and $\{v_{j_1}, \ldots, v_{j_t}\}$ are cliques, 
and therefore, $G$ contains $K_t\mat K_t$ as an induced subgraph, contradiction.
We conclude that $(T,L)$ is a linear branch-decomposition of linear mim-width at most $t-1$.
\end{proof}

\begin{corollary}\label{cor:maincocomparability}
Let $t\ge 2$ be an integer.
\begin{enumerate}[(1)]
\item Given an $n$-vertex $(K_t\mat K_t)$-free co-comparability graph and $\sigma, \rho \subseteq \mathbb{N}$,
one can solve \textsc{Min-(or Max-) $(\sigma, \rho)$-domination}  in time $\mathcal{O}(n^{3d(t-1)+4})$ where $d=\max( d(\sigma), d(\rho))$.
\item Given an $n$-vertex $(K_t\mat K_t)$-free co-comparability graph and a $(q\times q)$-matrix $D_q$,
one can solve \textsc{$D_q$-partitioning} in time $\mathcal{O}(qn^{3dq(t-1)+4})$ where $d=\max_{i,j} d(D_q[i,j])$.
\end{enumerate}
\end{corollary}

Note that permutation graphs do not contain $K_3\mat K_3$ as induced subgraphs, because $K_3\mat K_3$ is the complement of $C_6$, which is not a permutation graph.
Therefore, Corollary~\ref{cor:maincocomparability} generalizes the algorithms for permutation graphs.

In the next, we show that co-comparability graphs have unbounded mim-width.
For positive integers $p$ and $q$, 
the \emph{$(p\times q)$ column-clique grid} is  the graph on the vertex set $\{v_{i,j}:1\le i\le p,1\le j\le q\}$
where 
\begin{itemize}
\item for every $j\in \{1, \ldots, q\}$, $\{v_{1, j}, \ldots, v_{p, j}\}$ is a clique,
\item for every $i\in \{1, \ldots, p\}$ and $j_1, j_2\in \{1, \ldots, q\}$, $v_{i, j_1}$ is adjacent to $v_{i, j_2}$ if and only if $\abs{j_2-j_1}=1$, 
\item for $i_1, i_2\in \{1, \ldots, p\}$ and $j_1, j_2\in \{1, \ldots, q\}$ with $i_1\neq i_2$ and $j_1\neq j_2$, $v_{i_1, j_1}$ is not adjacent to $v_{i_2, j_2}$.
\end{itemize}
We depict an example in Figure~\ref{fig:cliquegrid}.
For each $1\le i\le p$, we call $\{v_{i,1}, \ldots, v_{i,h}\}$ \emph{the $i$-th row} of $G$, and define its columns similarly.

\begin{figure}[t]
\centerline{\includegraphics[scale=0.4]{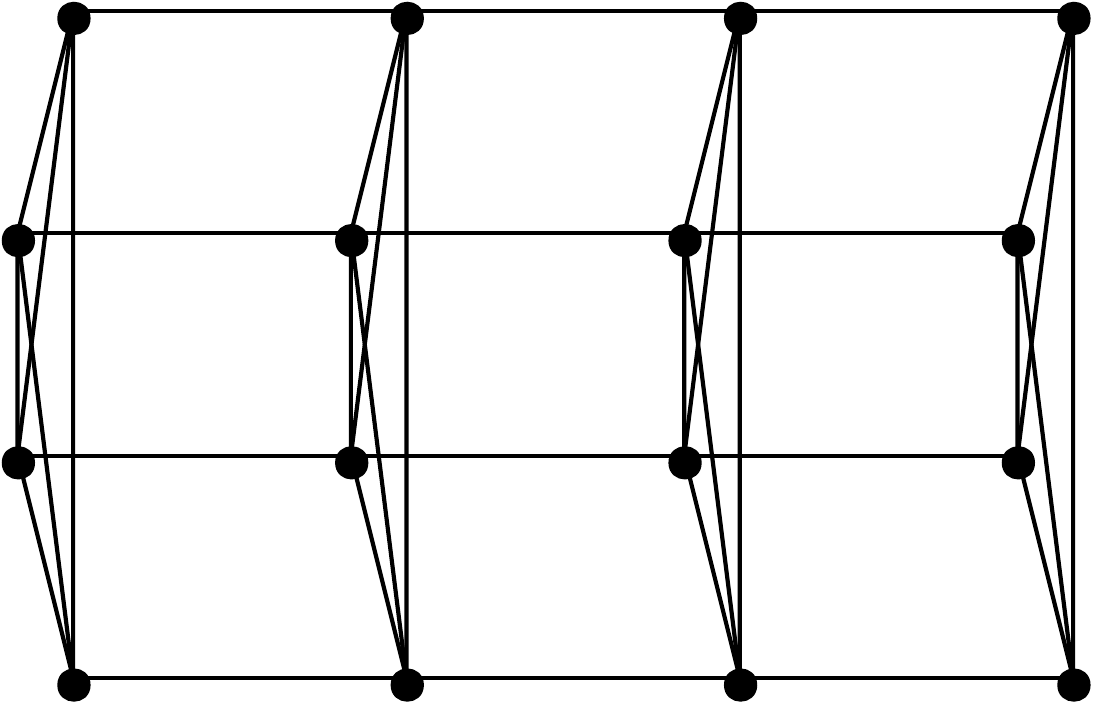}}
\caption{The $(4\times 4)$ column-clique grid.}
\label{fig:cliquegrid}
\end{figure}

\begin{lemma}\label{lem:lowerbdcocomparability}
For  integers $p, q\ge 12$, 
the $(p\times q)$ column-clique grid has mim-width at least $ \min (\frac{p}{4}, \frac{q}{3})$.
\end{lemma}
\begin{proof}
Let $G$ be the $(p\times q)$ column-clique grid, and let $(T,L)$ be a branch-decomposition of $G$.
It is enough to show that $(T, L)$ has mim-width at least $\min (\frac{p}{4}, \frac{q}{3})$.
Let $w$ be the mim-width of $(T,L)$.
By Lemma~\ref{lem:balancedpartition}, 
$G$ admits a vertex bipartition $(A, B)$ where $\mimval_G(A)\le w$ and for each $U\in \{A, B\}$, $\frac{\abs{V(G)}}{3}<\abs{U}\le \frac{2\abs{V(G)}}{3}$.
It is sufficient to show that $\mimval_G(A)\ge \min (\frac{p}{4}, \frac{q}{3})$.

Firstly, assume that for each row $R$ of $G$, $R\cap A\neq \emptyset$ and $R\cap B\neq \emptyset$.
Then there is an edge between $R\cap A$ and $R\cap B$, as $G[R]$ is connected.
For each $i$-th row $R_i$, we choose a pair of vertices $v_{i,a_i}\in R\cap A$ and $v_{i,b_i}\in R\cap B$ that are adjacent.
We choose a subset $X\subseteq \{1, \ldots, p\}$ such that $\abs{X}\ge \frac{p}{2}$ and
 every pair $(v_{i, a_i}, v_{i, b_i})$ in $\{(v_{i, a_i}, v_{i, b_i}):i\in X\}$ satisfies that $a_i+1=b_i$.
By taking the same parity of $a_i$'s, 
we choose a subset $Y\subseteq X$ such that $\abs{Y}\ge \frac{p}{4}$ and all integers in $\{a_i: i\in Y\}$ have the same parity.
Then we can observe that $\{v_{i, a_i}v_{i, b_i}:i\in Y\}$  is an induced matching in $G[A, B]$, as 
all integers in $\{a_i: i\in Y\}$ have the same parity.
Therefore, we have $\mimval_G(A)\ge \frac{p}{4}$.

Now, we assume that there exists a row $R$ such that $R$ is fully contained in one of $A$ and $B$.
Without loss of generality, we may assume that $R$ is contained in $A$.
Since $\abs{B}> \frac{\abs{V(G)}}{3}$, there is a subset $X\subseteq \{1, \ldots, q\}$ such that
$\abs{X}>\frac{q}{3}$ and for each $i\in X$, the $i$-th column contains a vertex of $B$.
For each $i$-th column where $i\in X$, we choose a vertex $v_{a_i, i}$ in $B$.
It is not hard to verify that the edges between $\{v_{a_i, i}:i\in X\}$ and the rows in $R$ form an induced matching of size $\frac{q}{3}$ in $G[A, B]$.

Therefore, we have $d\ge \min (\frac{p}{4}, \frac{q}{3})$.
\end{proof}

\begin{corollary}\label{cor:lowercocomp}
For every large enough $n$, there is a co-comparability graph on $n$ vertices having mim-width at least $\sqrt{\frac{n}{12}}$.
\end{corollary}

\begin{proof}
Let $p\ge 4$ be an integer, and let $n:=12p^2$.
Let $G$ be the $(4p\times 3p)$ column-clique grid.
It is not hard to see that
\[v_{1,1}, v_{2,1}, \ldots, v_{4p,1}, v_{1, 2}, v_{2, 2}, \ldots, v_{4p-1, 3p}, v_{4p, 3p}  \]
is a co-comparability ordering, as each column is a clique. Thus, $G$ is a co-comparability graph.
By Lemma~\ref{lem:lowerbdcocomparability}, 
$\mimw(G)\ge p=\sqrt{\frac{n}{12}}$.
\end{proof}

We remark that the two classes, $(K_t\mat S_t)$-free chordal graphs and $(K_t\mat K_t)$-free co-comparability graphs, are 
subclasses of the class of graphs of sim-width at most $1$ and having no induced subgraph isomorphic to $K_t\mat K_t$ and $K_t\mat S_t$ for $t\ge 3$.
This is because chordal graphs have no $K_2\mat K_2$ induced subgraph, and co-comparability graphs have no $K_3\mat S_3$ induced subgraphs by Lemma~\ref{lem:obscocom}.
Motivated from it, we will extend these classes to classes graphs of bounded sim-width and having no $K_t\mat K_t$ and $K_t\mat S_t$ as induced subgraphs or induced minors, in Section~\ref{sec:excludingtmatching}.

\section{Graphs of bounded sim-width}\label{sec:excludingtmatching}

In Section~\ref{sec:specialclass}, we proved that graphs of sim-width at most $1$ contain all chordal graphs and all co-comparability graphs.
A classical result on chordal graphs is that \textsc{Minimum Dominating Set} is NP-complete on chordal graphs~\cite{BoothJ1982}.
So, even for this kind of locally-checkable problems, we cannot expect efficient algorithms on graphs of sim-width at most $w$.
Therefore, to obtain a meta-algorithm for graphs of bounded sim-width encompassing many locally-checkable problems, we must impose some restrictions. We approach this problem in a way analogous to 
what has previously been done in the realm of rank-width \cite{FominOT2010}.

It is well known that complete graphs have rank-width at most $1$, but they have unbounded tree-width.
Fomin, Oum, and Thilikos~\cite{FominOT2010} showed that 
if a graph $G$ is $K_r$-minor free, then its tree-width is bounded by $c\cdot \rw(G)$ where $c$ is a constant depending on $r$.
This can be utilized algorithmically, to get a result for graphs of bounded rank-width when excluding a fixed minor, as the class of problems solvable in FPT time is strictly larger when parameterized by tree-width than rank-width \cite{HOSG2006}.

We will do something similar by focusing on the distinction between mim-width and sim-width.
However, $K_r$-minor free graphs are too strong, as one can show that on $K_r$-minor free graphs, 
the tree-width of a graph is also bounded by some constant factor of its sim-width. 
To see this, one can use Lemma~\ref{lem:contraction} and the result on contraction obstructions for graphs of bounded tree-width~\cite{FominGT2011}.

Instead of using minors, we exclude $K_t\mat K_t$ and $K_t\mat S_t$ as induced subgraphs or induced minors. The induced minor operation is rather natural because sim-width does not increase when taking induced minors; we prove this property in Subsection~\ref{subsec:inducedminor}. In Subsection~\ref{subsec:unboundedrank}, 
we prove that chordal graphs having no induced minor isomorphic to $K_3\mat S_3$ have unbounded rank-width.
As chordal graphs have no induced minor isomorphic to $K_3\mat K_3$ and they have sim-width $1$, 
this implies that graphs that have bounded sim-width and have no induced minor isomorphic to $K_t\mat K_t$ or $K_t\mat S_t$ for fixed $t$ have unbounded rank-width. Therefore, such classes still extend classes of bounded rank-width.

 We denote by $R(k,\ell)$ the \emph{Ramsey number}, that is the minimum integer satisfying that every graph with at least $R(k,\ell)$ vertices contains either a clique of size $k$ or an independent set of size $\ell$. By Ramsey's Theorem~\cite{Ramsey1930}, $R(k, \ell)$ exists for every pair of positive integers $k$ and $\ell$.

\subsection{Bounding mim-width}

We show the following.

\begin{proposition}\label{prop:boundingmimwidth}
Every graph with sim-width $w$ and no induced minor isomorphic to $K_t\mat K_t$ and $K_t\mat S_t$ has mim-width at most $8(w+1)t^3-1$.
\end{proposition}

\begin{proposition}\label{prop:boundingmimwidth2}
Every graph with sim-width $w$ and no induced subgraph isomorphic to $K_t\mat K_t$ and $K_t\mat S_t$ has mim-width at most $R( R(w+1,t), R(t,t))$.
\end{proposition}
We remark that sometimes this Ramsey number can go down to a polynomial function depending on the underlying graphs. 
See discussions in Belmonte et al~\cite{BelmonteHHRS2014}.

We first prove Proposition~\ref{prop:boundingmimwidth}.
We use the following result. Notice that the optimal bound of Theorem~\ref{thm:largeminor} has been slightly improved by Fox~\cite{Fox2010}, and then by Balogh and Kostochka~\cite{BaloghK2011}.

\begin{theorem}[Duchet and Meyniel~\cite{DuchetM1982}]\label{thm:largeminor}
For positive integers $k$ and $n$, every $n$-vertex graph contains either an independent set of size $k$ or 
a $K_t$-minor where $t\ge \frac{n}{2k-1}$. 
\end{theorem}

\begin{proof}[Proof of Proposition~\ref{prop:boundingmimwidth}]
Let $G$ be a graph with sim-width $w$ and no induced minor isomorphic to $K_t\mat K_t$ and $K_t\mat S_t$.
Let $(T,L)$ be a branch-decomposition of $G$ of sim-width $w$.
Let $e$ be an edge of $T$ and $(A,B)$ be the vertex bipartition of $G$ associated with $e$.
We claim that $\mimval_G(A)\le 8(w+1)t^3-1$.

Suppose for contradiction that there is an induced matching $\{v_1w_1, \ldots, v_mw_m\}$ in $G[A,B]$ where $v_1, \ldots, v_m\in A$, $w_1, \ldots, w_m\in B$, and $m\ge 8(w+1)t^3$.
Let $f$ be the bijection from $\{v_1, \ldots, v_m\}$ to $\{w_1, \ldots, w_m\}$ such that $f(v_i)=w_i$ for each $i\in \{1, \ldots, m\}$.
As $m\ge 8(w+1)t^3$, by Theorem~\ref{thm:largeminor},
the subgraph $G[\{v_1, \ldots, v_m\}]$ contains either an independent set of size $2(w+1)t$, or a $K_{2t^2}$-minor.

First assume that $G[\{v_1, \ldots, v_m\}]$ contains a $K_{2t^2}$-minor. 
Then there is a minor model of $K_{2t^2}$ in $G[\{v_1, \ldots, v_m\}]$, that is, 
there exist pairwise disjoint subsets $S_1, \ldots, S_{2t^2}$ of $\{v_1, \ldots, v_m\}$ such that
\begin{itemize}
\item for each $i\in \{1, \ldots, 2t^2\}$, $G[S_i]$ is connected, and
\item for two distinct integers $i,j\in \{1, \ldots, 2t^2\}$, there is an edge between $S_i$ and $S_j$.
\end{itemize}
For each $i\in \{1, \ldots, 2t^2\}$, we choose a representative $d_i$ in each $f(S_i)$ and contract each $S_i$ to a vertex $c_i$. Let $G_1$ be the resulting graph.
Then $G_1[\{c_1, \ldots, c_{2t^2}\}, \{d_1, \ldots, d_{2t^2}\}]$ is an induced matching of size $2t^2$, and $\{c_1, \ldots, c_{2t^2}\}$ is a clique in $G_1$.
By the same procedure on $\{d_1, \ldots, d_{2t^2}\}$, 
the subgraph $G_1[\{d_1, \ldots, d_{2t^2}\}]$ contains either an independent set of size $t$, or a $K_{t}$-minor. In both cases, 
one can observe that $G_1$ contains an induced minor isomorphic to $K_t\mat K_t$ or $K_t\mat S_t$, hence we obtain a contradiction.

Now assume that $G[\{v_1, \ldots, v_m\}]$ contains an independent set $\{c_1, \ldots, c_{2(w+1)t}\}$, and for each $i\in \{1, \ldots, 2(w+1)t\}$, let $d_i:=f(c_i)$.
By Theorem~\ref{thm:largeminor}, $G[\{d_1, \ldots, d_{2(w+1)t}\}]$ contains either an independent set of size $w+1$ or a $K_t$-minor.
In the former case, we obtain an induced matching of size $w+1$ between $A$ and $B$ in $G$, contradicting to the assumption that $\simval_{G}(A)\le w$.
In the latter case, we obtain an induced minor isomorphic to $K_t\mat S_t$, contradiction.

We conclude that $\mimval_{G}(A)\le 8(w+1)t^3-1$. Since $e$ is arbitrary, $G$ has mim-width at most $8(w+1)t^3-1$.
\end{proof}

In a similar manner we can prove Proposition~\ref{prop:boundingmimwidth2}.

\begin{proof}[Proof of Proposition~\ref{prop:boundingmimwidth2}]
One can easily modify from the proof of Proposition~\ref{prop:boundingmimwidth} by replacing the application of Theorem~\ref{thm:largeminor} with the Ramsey's Theorem to find a clique or an independent set.
\end{proof}

As a corollary, we obtain the following.
In general, we do not have a generic algorithm to find a decomposition, and thus we assume that the decomposition is given as an input.

\begin{corollary}\label{cor:main3}
Let $w\ge 1$ and $t\ge 2$ be integers.

\begin{enumerate}[(1)]
\item Given an $n$-vertex graph of sim-width at most $w$ and having induced minor isomorphic to neither $K_t\mat K_t$ nor $K_t\mat S_t$, 
 and $\sigma, \rho \subseteq \mathbb{N}$, one can solve \textsc{Min-(or Max-) $(\sigma, \rho)$-domination}  in time $\mathcal{O}(n^{3dt'+4})$ where $d=\max( d(\sigma), d(\rho))$ and $t'=8(w+1)t^3-1$.
\item Given an $n$-vertex graph of sim-width at most $w$ and having induced minor isomorphic to neither $K_t\mat K_t$ nor $K_t\mat S_t$, 
and a $(q\times q)$-matrix $D_q$,
one can solve \textsc{$D_q$-partitioning} in time $\mathcal{O}(qn^{3dqt'+4})$ where $d=\max_{i,j} d(D_q[i,j])$ and $t'=8(w+1)t^3-1$.
\item Given an $n$-vertex graph of sim-width at most $w$ and having induced subgraph isomorphic to neither $K_t\mat K_t$ nor $K_t\mat S_t$, 
 and $\sigma, \rho \subseteq \mathbb{N}$, one can solve \textsc{Min-(or Max-) $(\sigma, \rho)$-domination}  in time $\mathcal{O}(n^{3dt'+4})$ where $d=\max( d(\sigma), d(\rho))$ and $t'=R( R(w+1,t), R(t,t))$.
\item Given an $n$-vertex graph of sim-width at most $w$ and having induced subgraph isomorphic to neither $K_t\mat K_t$ nor $K_t\mat S_t$, 
and a $(q\times q)$-matrix $D_q$,
one can solve \textsc{$D_q$-partitioning} in time $\mathcal{O}(qn^{3dqt'+4})$ where $d=\max_{i,j} d(D_q[i,j])$ and $t'=R( R(w+1,t), R(t,t))$.
\end{enumerate}
\end{corollary}

\subsection{Induced minors}\label{subsec:inducedminor}

We show that the sim-width of a graph does not increase when taking an induced minor.
This is one of the main motivations to consider this parameter.

\begin{lemma}\label{lem:contraction}
The sim-width of a graph does not increase when taking an induced minor.
\end{lemma}

\begin{proof}
Clearly, the sim-width of a graph does not increase when removing a vertex.
We prove for contractions.

Let $G$ be a graph, $v_1v_2\in E(G)$, and let $(T, L)$ be a branch-decomposition of $G$ of sim-width $w$ for some positive integer $w$.
Let $z$ be the contracted vertex in $G/v_1v_2$.
We claim that $G/v_1v_2$ admits a branch-decomposition of $G$ of sim-width at most $w$.
We may assume that $G$ is connected and has at least $3$ vertices.
For $G/v_1v_2$, we obtain a branch-decomposition $(T',L')$ as follows:
\begin{itemize}
\item Let $T'$ be the tree obtained from $T$ by removing $L(v_2)$, and smoothing its neighbor. 
\item Let $L'$ be the function from $V(G/v_1v_2)$ to the set of leaves of $T'$ such that 
$L'(v)=L(v)$ for $v\in V(G/v_1v_2)\setminus \{z\}$ and $L'(z)=L(v_1)$.
\end{itemize}
Let $e_1$ and $e_2$ be the two edges of $T$ incident with the neighbor of $L(v_2)$, but not incident with $L(v_2)$.
Let $e_{cont}$ be the edge of $T'$ obtained by smoothing. By construction, all edges of $E(T')\setminus \{e_{cont}\}$ are contained in $T$.

\begin{claim}\label{claim:contraction}
For each $e\in E(T')\setminus \{e_{cont}\}$, the {$\simval_{G/v_1v_2}$-width} of $e$ in $(T', L')$ is at most the {$\simval_{G}$-width} of $e$ in $(T,L)$.
\end{claim}
\begin{clproof}
Let $e\in E(T')\setminus \{e_{cont}\}$.
Let $(A, B)$ be the vertex bipartition of $G/v_1v_2$ associated with $e$.
Without loss of generality, we may assume $z\in A$.
Suppose  there exists
an induced matching $\{a_1b_1, \ldots, a_mb_m\}$ in $G/ v_1v_2$ with $a_1, \ldots, a_m\in A$ and $b_1,\ldots, b_m\in B$.
Let $(A', B')$ be the vertex bipartition of $G$ associated with $e$ where $v_1\in A'$.
We will show that there is also an induced matching in $G$ of same size between $A'$ and $B'$. 

We have  
either $v_2\in A'$ or $v_2\in B'$.
If $z\notin \{a_1, \ldots, a_m\}$, then 
$\{a_1b_1, \ldots, a_mb_m\}$ is also an induced matching between $A'$ and $B'$ in $G$.
Without loss of generality, we may assume that $z=a_1$.

\medskip
\noindent\textbf{Case 1. $v_2\in A'$.}
\begin{proof}
In $G$, one of $v_1$ and $v_2$, say $v'$, is adjacent to $b_1$.
Also, $v_1$ and $v_2$ are not adjacent to any of $\{a_2, \ldots, a_m, b_1, \ldots, b_m\}$.
Therefore, $\{v'b_1, a_2b_2, \ldots, a_mb_m\}$ is an induced matching in $G$ between $A'$ and $B'$, as required.
\end{proof}

\medskip
\noindent\textbf{Case 2. $v_2\in B'$.}
\begin{proof}
In this case, $\{v_1v_2, a_2b_2, \ldots, a_mb_m\}$ is an induced matching between $A'$ and $B'$, because $v_1$ and $v_2$ are  not adjacent to 
any of $\{a_2, \ldots, a_m, b_2, \ldots, b_m\}$.
\end{proof}

It shows that the {$\simval_{G/v_1v_2}$-width} of $e$ in $(T', L')$ is at most the {$\simval_{G}$-width} of $e$ in $(T,L)$.
\end{clproof}

In a similar manner, we can show that the $\simval_{G/v_1v_2}$-width of $e_{cont}$ in $(T', L')$ is at most the minimum of the $\simval_{G/v_1v_2}$-width of $e_1$ and $e_2$ in $(T,L)$.
Thus, we conclude that $\simw(G/v_1v_2)\le \simw(G)$.
\end{proof}

 \subsection{Unbounded rank-width}\label{subsec:unboundedrank}

We show that the Hsu-clique chain graph depicted in Figure~\ref{fig:hsuclique} is chordal, but does not contain $K_2\mat K_2$ or $K_3\mat S_3$ as an induced minor. 
Belmonte and Vatshelle~\cite[Lemma 16]{BelmonteV2013} showed that a $(p\times q)$ Hsu-clique chain graph has rank-width at least $\frac{p}{3}$ when $q=3p+1$.
So, our algorithmic applications based on Proposition~\ref{prop:boundingmimwidth} or Proposition~\ref{prop:boundingmimwidth2} are beyond algorithmic applications of graphs of bounded rank-width.

We formally define Hsu-clique chain graphs.
For positive integers $p,q$, 
the \emph{$(p\times q)$ Hsu-clique chain grid} is  the graph on the vertex set $\{v_{i,j}:1\le i\le p,1\le j\le q\}$
where 
\begin{itemize}
\item for every $j\in \{1, \ldots, q\}$, $\{v_{1, j}, \ldots, v_{p, j}\}$ is a clique,
\item for every $i_1, i_2\in \{1, \ldots, p\}$ and $j\in \{1, \ldots, q-1\}$, $v_{i_1, j}$ is adjacent to $v_{i_2, j+1}$ if and only if $i_1\le i_2$, 
\item for $i_1, i_2\in \{1, \ldots, p\}$ and $j_1, j_2\in \{1, \ldots, q\}$, $v_{i_1, j_1}$ is not adjacent to $v_{i_2, j_2}$ if $\abs{j_1-j_2}>1$.
\end{itemize}

\begin{figure}[t]
\centerline{\includegraphics[scale=0.4]{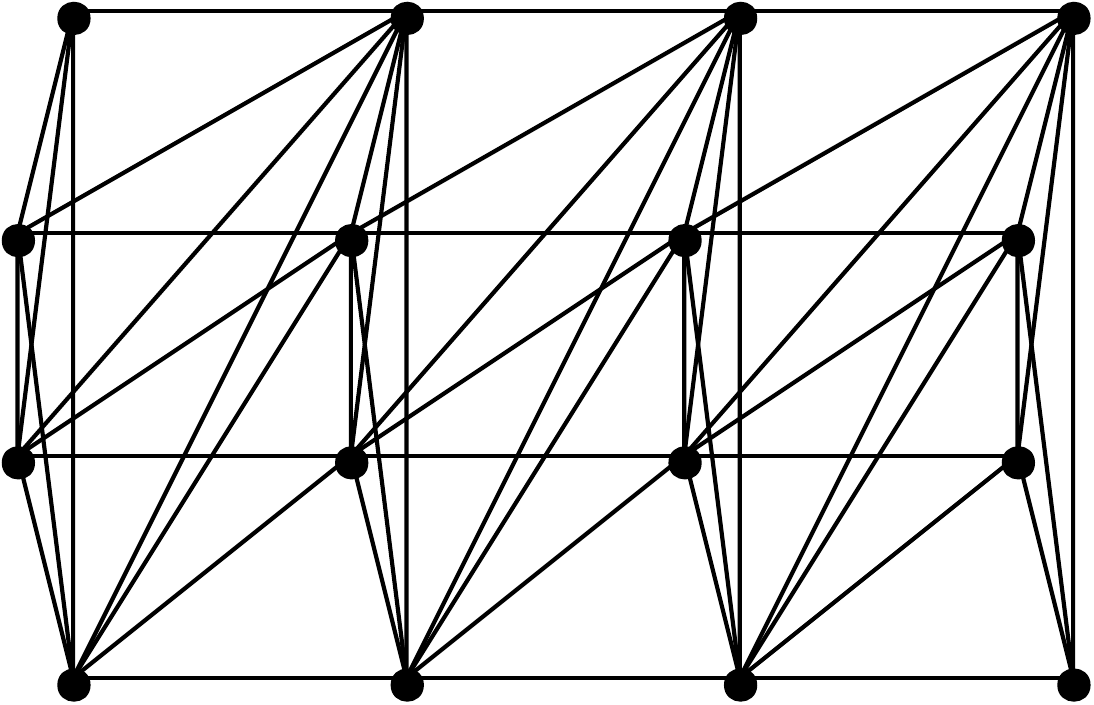}}
\caption{The $(4\times 4)$ Hsu-clique chain graph.}
\label{fig:hsuclique}
\end{figure}

\begin{proposition}\label{prop:nok3s3inducedminor}
Let $w\ge 1$ and $t\ge 2$ be integers.
The class of graphs of sim-width $w$ and having induced minor isomorphic to neither $K_t\mat S_t$ nor $K_t\mat K_t$ has unbounded rank-width.
\end{proposition}

\begin{proof}
For this, we provide that Hsu-clique chain grids are chordal and having no induced minor isomorphic to $K_t\mat S_t$ or $K_t\mat K_t$.
Since Hsu-clique chain grids have unbounded rank-width~\cite{BelmonteV2013}, it proves the proposition.

Let $G$ be the $(p\times q)$ Hsu-clique chain graph for some positive integers $p$ and $q$.

First show that $G$ is chordal. Suppose for contradiction that $G$ contains an induced cycle $C=c_1c_2 \cdots c_mc_1$ where $m\ge 4$.
Since each column of $G$ is a clique, all vertices of $C$ are contained in two consecutive columns.
But also, since each column contains at most $2$ vertices, we have $m=4$, and two consecutive vertices of $C$ are contained in one column.
Without loss of generality, we assume $c_1, c_2$ are in the $i$-th column for some $i\in \{1, \ldots, q-1\}$, 
and $c_3, c_4$ are in the $(i+1)$-th column. This implies that there is an induced matching of size $2$ between two columns if we ignore edges in each column. However, this is not possible from the construction. Therefore, $G$ is chordal.

If $G$ contains an induced minor isomorphic to $K_2\mat K_2$, which is an induced cycle of length $4$, 
then $G$ contains an induced subgraph isomorphic to an induced cycle of length $\ell$ for some $\ell\ge 4$.
This contradicts to the fact that $G$ is chordal.
So, $G$ has no induced minor isomorphic to $K_2\mat K_2$.

We claim that $G$ has no induced minor isomorphic to $K_3\mat S_3$.
For contradiction, suppose there is an induced minor model $\eta$ of $K_3\mat S_3$ in $G$.
Let $H:=K_3\mat S_3$.
For each $v\in V(H)$, let $I_v=\{j:v_{i,j}\in V(\eta(v))\}$. 
Let $I:=I_{v^1_1}\cup I_{v^1_2}\cup I_{v^1_3}$, and let $\ell$ and $r$ be the smallest and greatest integers in $I$, respectively.
Let $x,y\in \{v^1_1, v^1_2, v^1_3\}$ such that
\begin{enumerate}
\item $V(\eta(x))$ contains a vertex in the $\ell$-th column, but for $z\in \{v^1_1, v^1_2, v^1_3\}\setminus \{x\}$, $V(\eta(z))$ has no vertex whose row index is higher than all vertices in $V(\eta(x))$ in the $\ell$-th column, 
\item similarly, $V(\eta(y))$ contains a vertex in the $r$-th column, but for $z\in \{v^1_1, v^1_2, v^1_3\}\setminus \{x\}$, $V(\eta(z))$ has no vertex whose row index is lower than all vertices in $V(\eta(y))$ in the $r$-th column. 
\end{enumerate}
As $\{v^1_1, v^1_2, v^1_3\}$ is a clique, it is easy to observe that $I_x\cup I_y=I$.
Let $\ell$ be the integer such that $v^1_{\ell}\in \{v^1_1, v^1_2, v^1_3\}\setminus \{x,y\}$.
By the choice of $x$ and $y$, every vertex in $V(\eta(v^2_{\ell}))$ has a neighbor in $V(\eta(x))\cup V(\eta(y))$.
Thus, it contradicts to the assumption that $G$ contains $H$ as an induced minor.
\end{proof}

\section{Lower bound of sim-width on circle graphs}\label{sec:circlegraph}
In this section, we construct a set of circle graphs with unbounded sim-width and no triangle, using the approach similar to the one by Mengel~\cite{Mengel}. As $\simw(G) \leq \mimw(G)$ for every graph $G$, the example also implies that the class of circle graphs has unbounded mim-width.

\begin{theorem}\label{thm:unbddsimw}
Circle graphs have unbounded sim-width, and thus have unbounded mim-width.
\end{theorem}

For an integer $d \geq 0$, a graph $G$ is \emph{$d$-degenerate} if every subgraph of $G$ contains a vertex of degree at most $d$.
The following lemma asserts that every $d$-degenerate graph with large treewidth also has a large mim-width.

\begin{lemma}[Vatshelle's Thesis~\cite{VatshelleThesis}; See also Mengel~\cite{Mengel}]\label{lem:mimwtw}
For every $d$-degenerate graph $G$, $\mimw(G) \geq \frac{\tw(G)+1}{3(d+1)}$.
\end{lemma}

Given a bipartite graph with large mim-width, we do not lose much of mim-width after adding edges to one part of the bipartition.

\begin{lemma}[Mengel~\cite{Mengel}]\label{lem:bddmw}
Let $H$ be a bipartite graph with a bipartition $(A,B)$, and let $G$ be a graph obtained from $H$ by adding some edges in $H[A]$. Then $\mimw(G) \geq \frac{\mimw(H)}{2}$.
\end{lemma}

For an integer $k \geq 1$ and a graph $H$, a graph $G$ is a \emph{$k$-subdivision} of $H$ if $G$ can be built from $H$ by replacing every edge of $H$ with a path of length $k+1$ such that those paths replacing edges of $H$ are internally vertex-disjoint.

Now we prove the main lemma of this section.

\begin{figure}[t]
\centerline{\includegraphics[scale=0.65]{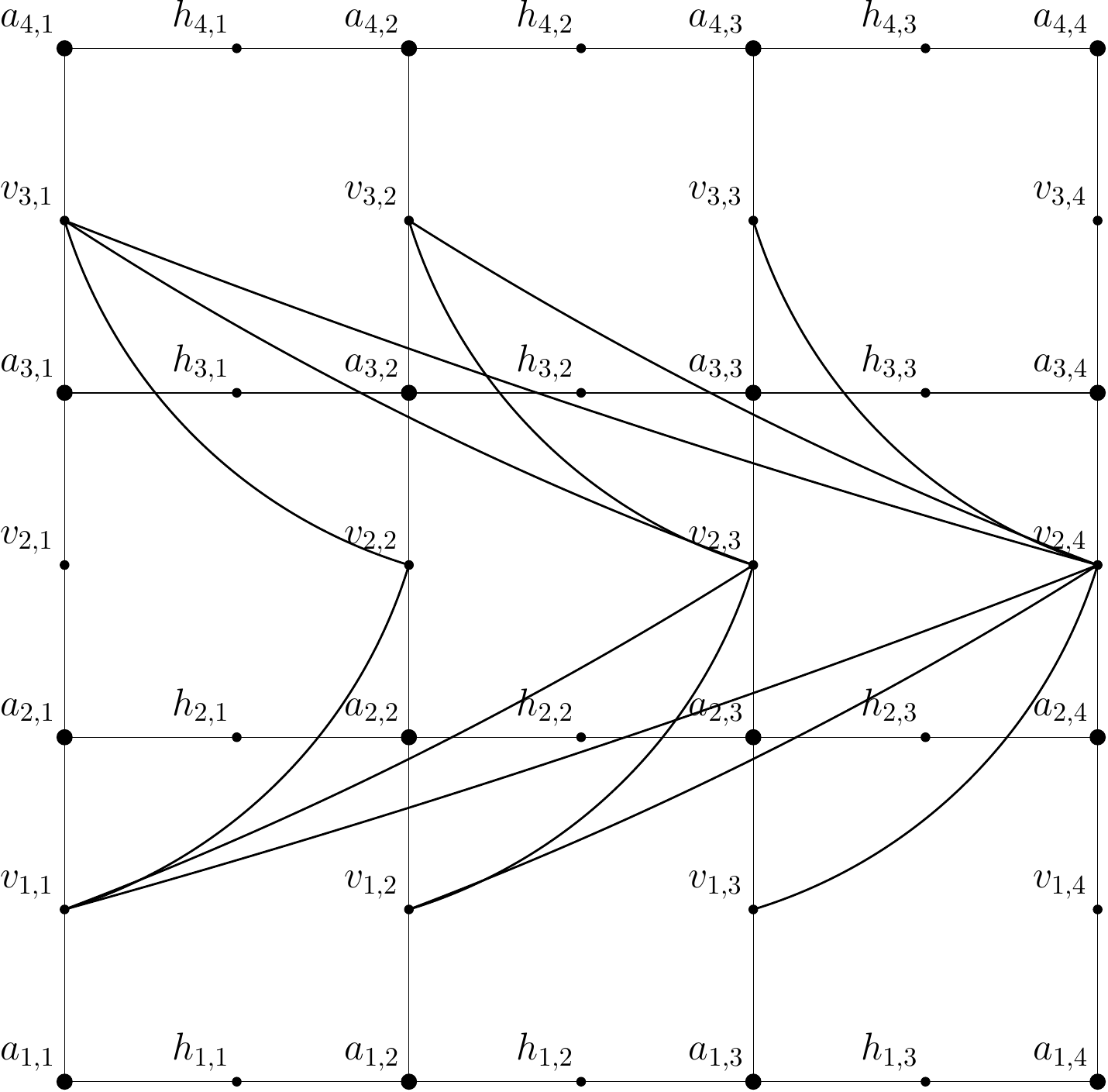}}
\caption{The circle graph $G_4$ of Lemma~\ref{lem:construction}}
\label{fig:circlegraph}
\end{figure}

\begin{lemma}\label{lem:construction}
For every integer $k \geq 2$, there is a circle graph $G_k$ satisfying the following.
\begin{enumerate}
\item $G_k$ contains no $K_3$ as a subgraph.
\item $G_k$ contains a 1-subdivision of a $(k \times k)$-grid as a subgraph. In particular, $\tw(G_k) \geq k$.
\item $\mimw(G_k) \geq \frac{k+1}{18}$.
\end{enumerate}
\end{lemma}
\begin{proof}
Let $H$ be a 1-subdivision of a $(k \times k)$-grid. It is easy to see that $H$ is 2-degenerate and has tree-width at least $k$, since it contains a $(k \times k)$-grid as a minor. By Lemma~\ref{lem:mimwtw}, $\mimw(H) \geq \frac{k+1}{9}$.

For $i,j \in \{1, \ldots, k\}$, let $a_{i,j}$ be the vertex on the $i$-th row and $j$-th column of the $(k \times k)$-grid. 
For $i\in \{1, \ldots, k\}$ and $j\in \{1, \ldots, k-1\}$, let $h_{i,j}$ be the vertex of degree $2$ adjacent to $a_{i,j}$ and $a_{i,j+1}$ in $H$. 
For $i\in \{1, \ldots, k-1\}$ and $j\in \{1, \ldots, k\}$, let $v_{i,j}$ be the vertex of degree $2$ adjacent to $a_{i,j}$ and $a_{i+1,j}$ in $H$.
Let us define
$A:=\{a_{i,j} : i,j\in \{1, \ldots, k\} \}$, 
$B_h:=\{ h_{i,j} : i\in \{1, \ldots, k\}, j\in \{1, \ldots, k-1\} \}$, and
$B_v:=\{ v_{i,j} : i\in \{1, \ldots, k-1\}, j\in \{1, \ldots, k\} \}$.
Then $H$ is a bipartite graph with the bipartition $(A, B_h\cup B_v)$. Let $B = B_h \cup B_v$. 

Before we describe in detail, we briefly explain how we will construct a circle representation of $H$.
For two points $a$ and $b$ on a circle, we denote by $\overline{ab}$ the chord whose end points are $a$ and $b$.
Let $p_1 , \ldots , p_{2k^2}$ be distinct $2k^2$ points on the circle in clockwise order. 
We regard each vertex in $A$ as a chord joining two points $p_{2i}$ and $p_{2i-1}$ for some $i\in \{1, \ldots, k^2\}$. 
Since every vertex in $B$ has a degree 2 in $H$, the chord representing a vertex in $B$ will be represented by a chord intersecting two disjoint chords $\overline{p_{2i} p_{2i-1}}$ and $\overline{p_{2j} p_{2j-1}}$ for some $i\neq j$ that represents its neighbors. 
Since the chords representing vertices in $B$ may cross each other, the resulting circle graph will be a graph obtained from $H$ by adding some edges in $B$. 
By Lemma~\ref{lem:bddmw}, this graph has mim-width at least $\frac{\mimw(H)}{2} \geq \frac{k+1}{18}$.

We explain how we can assign chords representing vertices in $A$ not to induce $K_3$ in $B$.

Let $C$ be a circle and $c_1 , \dots , c_{6k^2}$ be distinct $6k^2$ points on $C$ in clockwise order. 
For $i, j\in \{1, \ldots, k\}$, $a_{i,j}$ is represented as follows:
\begin{itemize}
\item if $i$ is odd, then $a_{i,j}$ is represented by a chord intersecting $c_{6(i-1)k + 6j-5}$ and $c_{6(i-1)k + 6j}$,  
\item if $i$ is even, then $a_{i,j}$ is represented by a chord intersecting $c_{6ik - 6j + 1}$ and $c_{6ik - 6j + 6}$. 
\end{itemize}
Remark that we assign chords representing vertices 
$a_{1,1}, \ldots , a_{1,k}, a_{2,k}, \ldots , a_{2,1}, \ldots$ in clockwise order.
We assign chords for vertices in $B_h$ as follows.
\begin{itemize}
\item For an odd integer $i\in \{1, \ldots, k\}$ and an integer $j\in \{1, \ldots, k-1\}$, a chord connecting $c_{6(i-1)k + 6j-1}$ and $c_{6(i-1)k + 6j+2}$ represents a vertex $h_{i,j} \in B_h$.
\item For an even integer $i\in \{1, \ldots, k\}$ and an integer $j\in \{1, \ldots, k-1\}$, a chord connecting $c_{6ik - 6j + 2}$ and $c_{6ik - 6j - 1}$ represents a vertex $h_{i,j} \in B_h$.
\end{itemize}
We assign chords for vertices in $B_v$ as follows.
Note that we assign chords to vertices in $B_v$ so that for $i\in \{1, \ldots, k-2\}$ and $j_1, j_2\in \{1, \ldots, k\}$, 
chords representing two vertices $v_{i,j_1}$ and $v_{i+1,j_2}$ cross if and only if either $i$ is odd and $j_1 < j_2$ or $i$ is even and $j_1 > j_2$.

\begin{itemize}
\item For an odd integer $i\in \{1, \ldots, k-1\}$ and an integer $j\in \{1, \ldots, k\}$, a chord connecting $c_{6(i-1)k + 6j-2}$ and $c_{6ik + 6(k-j)+3}$ represents a vertex $v_{i,j} \in B_v$.
\item For an even integer $i\in \{1, \ldots, k-1\}$ and an integer $j\in \{1, \ldots, k\}$, a chord connecting $c_{6ik - 6j + 4}$ and $c_{6ik + 6j - 3}$ represents a vertex $v_{i,j} \in B_v$.
\end{itemize}

Let $G_k$ be the circle graph obtained by this procedure. See Figure~\ref{fig:circlegraph} which describes $G_4$. Since every new edge of $G_k$ joins $v_{i,j_1}$ and $v_{i+1,j_2}$ for some $i\in \{1, \ldots, k-2\}$ and $j_1, j_2\in \{1, \ldots, k\}$, it is easy to see that $G_k$ contains no $K_3$. 
Note that $G_k$ is a graph obtained from $H$ by adding some edges in $B$.
By Lemma~\ref{lem:bddmw}, the resulting circle graph has mim-width at least $\frac{\mimw(H)}{2} \geq \frac{k+1}{18}$.
\end{proof}

\begin{proof}[Proof of Theorem~\ref{thm:unbddsimw}]
By Proposition~\ref{prop:boundingmimwidth2}, there is an increasing function $f : \mathbb{N} \to \mathbb{N}$ such that every graph $G$ with no induced subgraph isomorphic to $K_3 \mat K_3$ and $K_3 \mat S_3$ has mim-width at most $f(\simw(G))$. Since the circle graph $G_k$ of Lemma~\ref{lem:construction} contains no $K_3$ as a subgraph, we have $f(\simw(G_k)) \geq \mimw(G_k) \geq \frac{k+1}{18}$. It follows that $\simw(G_k)$ should be large for sufficiently large $k$.
\end{proof}

\section{Concluding remarks}\label{sec:conclusion}

In this paper, we showed that  every fixed LC-VSVP problem can be solved in XP time parameterized by $t$ on $(K_t\mat S_t)$-free chordal graphs and $(K_t\mat K_t)$-free co-comparability graphs.
We further give a mim-width bound on the class of graphs that have sim-width at most $w$ and do not contain $K_t \mat S_t$ and $K_t\mat K_t$ as induced minors or induced subgraphs.. 

Vatshelle~\cite{VatshelleThesis} asked whether there is a function $f$ such that given a graph $G$, one can in XP time parameterized by $\mimw(G)$ compute a branch-decomposition of mim-width at most $f(k)$.
This still remains as an open problem.  We asked the same question for sim-width.

\begin{QUE}
Does there exist a function $f$ such that, given a graph $G$, one can in XP time parameterized by the sim-width $k$ of $G$ compute a decomposition of sim-width $f(k)$?
\end{QUE}

We observed that one cannot expect XP algorithms for \textsc{Minimum Dominating Set} parameterized by sim-width.
However, we know that one can solve \textsc{Maximum Independent Set} or \textsc{$3$-Coloring} in polynomial time on both chordal graphs and co-comparability graphs, which are all known classes of constant sim-width. We ask whether those problems are NP-complete on graphs of sim-width $1$ or not.

\begin{QUE}
Is \text{Maximum Independent Set} NP-complete on graphs of sim-width at most $1$?
Also, is \text{$3$-Coloring} NP-complete on graphs of sim-width at most $1$?
\end{QUE}

It would be interesting to find more classes having constant sim-width, but unbounded  mim-width. We propose some possible classes, that are also presented in Figure~\ref{fig:diagram}.

\begin{QUE}
Do weakly chordal graphs or AT-free graphs have constant sim-width? 
\end{QUE}

We showed that \textsc{Minimum Dominating Set} can be solved in time $n^{\mathcal{O}(t)}$ on $(K_t\mat S_t)$-free chordal graphs, but we could not obtain an FPT algorithm. We ask whether it is W[1]-hard or not. 

\begin{QUE}
Is \textsc{Minimum Dominating Set} on chordal graphs W[1]-hard parameterized by the maximum $t$ such that the given graph has an $K_t\mat S_t$ induced subgraph? 
\end{QUE}

\end{document}